\documentclass[11pt]{article}
\usepackage{natbib}
\usepackage{graphicx}
\usepackage{multirow}
\usepackage{amsmath,amssymb,amsfonts}
\usepackage{amsthm}
\usepackage{mathrsfs}
\usepackage[title]{appendix}
\usepackage{xcolor}
\usepackage{textcomp}
\usepackage{manyfoot}
\usepackage{booktabs}
\usepackage{algorithm}
\usepackage{algpseudocode}
\usepackage{listings}
\usepackage{a4wide}
\usepackage{bbm}
\usepackage{bm}
\usepackage{subcaption}
\usepackage{authblk}   

\newcommand{\E}{\mathbb{E}}
\newcommand{\e}{\mathrm{e}}
\newcommand{\Var}{\mathrm{Var}}

\newcommand{\D}{\mathrm{d}}
\newcommand{\dmsig}[1]{\makebox[1.25em][l]{\textsuperscript{#1}}}
\DeclareMathOperator*{\argmax}{arg\,max}
\DeclareMathOperator*{\argmin}{arg\,min}

\theoremstyle{plain}
\newtheorem{theorem}{Theorem}
\newtheorem{lemma}[theorem]{Lemma}
\newtheorem{corollary}[theorem]{Corollary}
\newtheorem{proposition}[theorem]{Proposition}

\theoremstyle{definition}
\newtheorem{assumption}{Assumption}
\newtheorem{definition}{Definition}

\theoremstyle{remark}
\newtheorem{remark}[theorem]{Remark}

\begin{document}
	
\title{Long-memory GARCH via a two-dimensional Markov chain}
\author[1]{Kyungsub Lee}
\author[1]{Kennedy Titus Kayaki}

\affil[1]{Department of Statistics, Yeungnam University,
	Gyeongsan, Republic of Korea}
	
\maketitle

\begin{abstract}
	This paper proposes a GARCH-type volatility model in which
	level-and-slope updates of a latent power-law kernel generate
	state-dependent decay of past shocks within a two-dimensional Markov
	state.  We derive a joint Foster--Lyapunov condition and establish
	positive Harris recurrence and uniqueness of the invariant
	distribution.  
	Simulations show substantial low-frequency persistence
	in log-squared innovations, especially near the diagnostic stability
	boundary.  
	Empirically, the model captures a substantial portion of
	observed volatility persistence and delivers competitive out-of-sample
	forecast accuracy using only a two-dimensional Markov state.
\end{abstract}

\section{Introduction}

This paper proposes a long-memory GARCH-type model admitting a finite-dimensional Markov representation 
and studies its stability.  
The model captures persistent volatility dynamics while
retaining a finite-dimensional Markov representation.  
We prove positive Harris recurrence of the resulting state process under
a joint stability condition that can be checked numerically from the
associated drift functions.

ARCH and GARCH models provide a basic framework for conditional
heteroskedasticity in financial returns
\citep{Engle1982,Bollerslev1986}.  In these models, conditional variance
depends on past squared innovations and, in the GARCH case, on past
conditional variances.  This feedback structure captures volatility
clustering in a simple recursive form.

Empirically, however, volatility persistence is often stronger than what
is implied by conventional GARCH specifications.  Volatility proxies such
as squared and absolute returns often exhibit slowly decaying dependence
\citep{Ding1993}.  This motivated extensions of the GARCH framework
aimed at stronger persistence.
Fractional GARCH models, including FIGARCH, FIEGARCH, and HYGARCH,
introduce fractional or hyperbolic decay into the volatility recursion
\citep{Baillie1996,BollerslevMikkelsen1996,Davidson2004}.  A related
approach is to use infinite-order ARCH dynamics, which extend
finite-order GARCH recursions by allowing longer lag structures
\citep{GiraitisKokoszkaLeipus2000}.

Long-memory volatility has also been studied outside the GARCH
framework.  Long-memory stochastic volatility models generate persistent
dependence through latent volatility factors with slowly decaying
autocovariances \citep{Breidt1998}.  More recently, rough-volatility
models have emphasized highly irregular volatility paths and
low-regularity stochastic drivers
\citep{GatheralJaissonRosenbaum2018,DiNunnoEtAl2023}.
Recent work also emphasizes the difficulty of empirically distinguishing
long-memory and rough-volatility specifications
\citep{LiPhillipsShiYu2025}.
In these approaches, persistence is generated
by a mechanism imposed on the recursion, whether a fractional filter, an
infinite-order structure, a latent long-memory component, or a rough
driver.

This paper takes a different approach.  It adapts the level-and-slope
matching construction of \citet{Lee2026} to conditional volatility
modeling.  The original construction updates power-law intensities while
preserving a finite-dimensional Markovian state.  Here, the same
mechanism is applied to a latent volatility kernel.  Each squared return
shock updates the kernel by matching its level and slope.  The resulting
state has two coordinates: one represents the variance level, and the
other represents the endogenous memory scale.

Related uses of slowly decaying or power-law structures appear in the
Hawkes-process literature; see, for example,
\citet{BacryDayriMuzy2012}, \citet{HardimanBercotBouchaud2013},
\citet{zhang2016modeling}, \citet{KanazawaSornette2020PRL}, and
\citet{kanazawa2021ubiquitous}.  Here the power-law structure is used not
as an intensity, but as a latent kernel that determines conditional
volatility.

The proposed model also differs from deterministic-coefficient
\(\mathrm{ARCH}(\infty)\) models.  Its effective decay factors are random
and endogenous, because they are determined by the evolving memory scale
rather than by fixed deterministic coefficients.  Persistence is
therefore generated by the dynamics of the state itself.

The main mathematical contribution is a stability theory for this
Markovian long-memory volatility model.  The memory scale evolves with
the shocks, so stability does not reduce to the one-dimensional
log-moment condition of standard GARCH models.  Instead,
the paper derives a Foster--Lyapunov drift condition for a Lyapunov
function that couples the variance-level coordinate with the memory-scale
coordinate.  Under this condition, together with regularity assumptions
on the innovation distribution, the Markov chain is positive Harris
recurrent and admits a unique invariant probability measure.

The argument is based on the general-state-space Markov-chain framework
of \citet{meyn2009markov}.  Related uses of this approach include
nonlinear autoregressive models
\citep{AnHuang1996,DoucFortMoulinesSoulier2004}, GARCH-type models
\citep{Cline2007,MeitzSaikkonen2008,BoussamaFuchsStelzer2011,
	MeitzSaikkonen2025}, discrete-time Hawkes processes
\citep{CostaMaillardMuraro2024}, and self-exciting point-process models
\citep{Lee2026FRPP}.

The paper also examines the persistence generated by the model through
simulation.  The Markov recursion admits a random-coefficient
\(\mathrm{ARCH}(\infty)\) representation with endogenous decay factors.
When the memory scale is large, these decay factors remain close to one
over long horizons, producing slowly decaying dependence in volatility
proxies.
We use the local Whittle estimator of \citet{robinson1995gaussian},
treating log-squared returns as volatility proxies following
\citet{Ding1993,Breidt1998}.  Applied to simulated series, it indicates
substantial low-frequency persistence, particularly near the stability
frontier.

We also evaluate the model empirically using data from international
financial markets.  
The fitted model generates substantial long-memory
behavior in log-squared returns and closely matches the estimated
memory for several markets.  
Its in-sample likelihood is lower than that of FIGARCH in each of the six markets considered 
and lower than that of the asymmetric GARCH models for the five equity indices.  
In contrast, its out-of-sample performance is closer to that of FIGARCH: 
although FIGARCH has a lower average QLIKE loss across the six markets considered, 
the difference is statistically significant only for KOSPI and Bitcoin.  
Its QLIKE loss is also similar to that of GARCH((1,1)) in most markets.  
HAR-RV, which uses lagged realized variance directly, provides the most accurate forecasts overall.

The remainder of the paper is organized as follows. 
Section~\ref{Sec:model} defines the model and derives the two-dimensional
recursion from the level-and-slope matching construction. 
Section~\ref{sec:Markov} formulates the recursion as a Markov chain and
proves the Foster--Lyapunov drift condition under the joint stability
assumption. 
Section~\ref{sec:regularity} verifies the control-chain regularity
conditions and establishes positive Harris recurrence. 
Section~\ref{sec:long-memory} studies the model's low-frequency
persistence by simulation. 
Section~\ref{sect:empirical} examines its empirical long-memory properties and compares its in-sample fit and out-of-sample forecasting performance with standard volatility benchmarks. 
Section~\ref{sec:conclusion} concludes the paper.

\section{Model} \label{Sec:model}
	
We consider a discrete-time model for a financial return series \(\{\eta_n\}_{n\ge1}\) 
and its conditional variance process \(\{\sigma_n^2\}_{n\ge1}\).  
The variable \(\eta_n\) denotes the return, or mean-adjusted innovation, realized over the interval
\([t_{n-1},t_n]\), and
\[
\sigma_n^2
:=
\Var(\eta_n\mid\mathcal F_{n-1})
\]
is its conditional variance.  Thus \(\sigma_n^2\) is
\(\mathcal F_{n-1}\)-measurable, as in standard GARCH-type models.

In contrast to standard GARCH recursions with a fixed autoregressive
decay, we generate the conditional variance from a latent power-law
kernel.  After observing \(\eta_n\) at time \(t_n\), we associate with
the next period \([t_n,t_{n+1}]\) the kernel
\begin{equation}
	h_n(t)
	=
	\mu+\frac{a_n}{(t+c_n)^p},
	\qquad t\ge0,
	\label{Eq:h}
\end{equation}
where \(\mu>0\) is the long-run baseline level, \(p>1\) is the
power-law decay exponent, and \((a_n,c_n)\) are state variables known
at time \(t_n\).  The argument \(t\) is a forward time variable,
measuring elapsed time from \(t_n\) toward the next observation time.

The conditional variance of the next innovation is defined as the
initial value of this updated kernel:
\begin{equation}
	\sigma_{n+1}^{2}
	:=
	h_n(0)
	=
	\mu+\frac{a_n}{c_n^p}
	=:
	\mu+X_n,
	\label{eq:sigma}
\end{equation}
where
\[
X_n:=\frac{a_n}{c_n^p}
\]
is the excess-variance state after the \(n\)-th innovation has been
incorporated.  Thus the pair \((X_n,c_n)\) represents the post-\(n\)
volatility state: \(X_n\) fixes the next conditional variance, while
\(c_n\) controls the future decay rate of the latent kernel.

With this notation, the timing convention is
\begin{equation}
	\underbrace{(X_{n-1},\, c_{n-1})}_{\mathcal F_{n-1}}
	\;\longrightarrow\;
	\underbrace{\sigma_n^{2} = \mu + X_{n-1}}_{\text{variance set}}
	\;\longrightarrow\;
	\underbrace{\eta_n \mid \mathcal F_{n-1}
		\sim (0,\,\sigma_n^{2})}_{\text{innovation drawn}}
	\;\longrightarrow\;
	\underbrace{(X_n,\, c_n)}_{\text{state updated}}.
	\label{eq:model-timing}
\end{equation}
Here the notation
\(\eta_n\mid\mathcal F_{n-1}\sim(0,\sigma_n^2)\) specifies only the
conditional mean and variance; the precise law is introduced below
through the normalized squared innovation.  Thus \(X_{n-1}\) determines the
conditional variance of \(\eta_n\), whereas \(X_n\), after incorporating
the realized shock \(\eta_n^2\), determines the next conditional
variance \(\sigma_{n+1}^2\).
Figure~\ref{fig:kernel-timing} illustrates this update.

\paragraph{Update rules.}
The following update rule is based on the level-and-slope matching
construction introduced in \citet{Lee2026}.
The state variables \((a_n,c_n)\) are updated at time \(t_n\), after the
innovation \(\eta_n\) has been realized.  
We take \(\tau:=t_n-t_{n-1}>0\) to be a constant inter-observation time.  
During the interval \([t_{n-1},t_n]\), the previous kernel \(h_{n-1}\) decays from
\(h_{n-1}(0)\) to \(h_{n-1}(\tau)\).  
At time \(t_n\), the realized shock
\(\eta_n^2\) contributes a new power-law component
\[
\frac{\alpha\eta_n^2}{(t+\gamma)^p},
\]
where \(\alpha>0\) controls the shock amplitude and \(\gamma>0\) is the
initial decay offset.  The updated kernel \(h_n\) is chosen within the
same two-parameter family by matching both the level and the slope at
the new origin \(t=0\).

\begin{itemize}
	\item \textbf{Level matching.}  The initial level of the new kernel
	equals the decayed level of the previous kernel plus the level of the
	new shock component:
	\[
	h_n(0)
	=
	h_{n-1}(\tau)
	+
	\frac{\alpha}{\gamma^p}\eta_n^2.
	\]
	
	\item \textbf{Slope matching.}  The initial slope of the new kernel
	equals the decayed slope of the previous kernel plus the initial slope
	of the new shock component:
	\[
	\left.\frac{\partial h_n}{\partial t}\right|_{t=0}
	=
	\left.\frac{\partial h_{n-1}}{\partial t}\right|_{t=\tau}
	+
	\left.
	\frac{\partial}{\partial t}
	\left(
	\frac{\alpha\eta_n^2}{(t+\gamma)^p}
	\right)
	\right|_{t=0}.
	\]
\end{itemize}

\begin{figure}[t]
	\centering
	\includegraphics[width=0.9\linewidth]{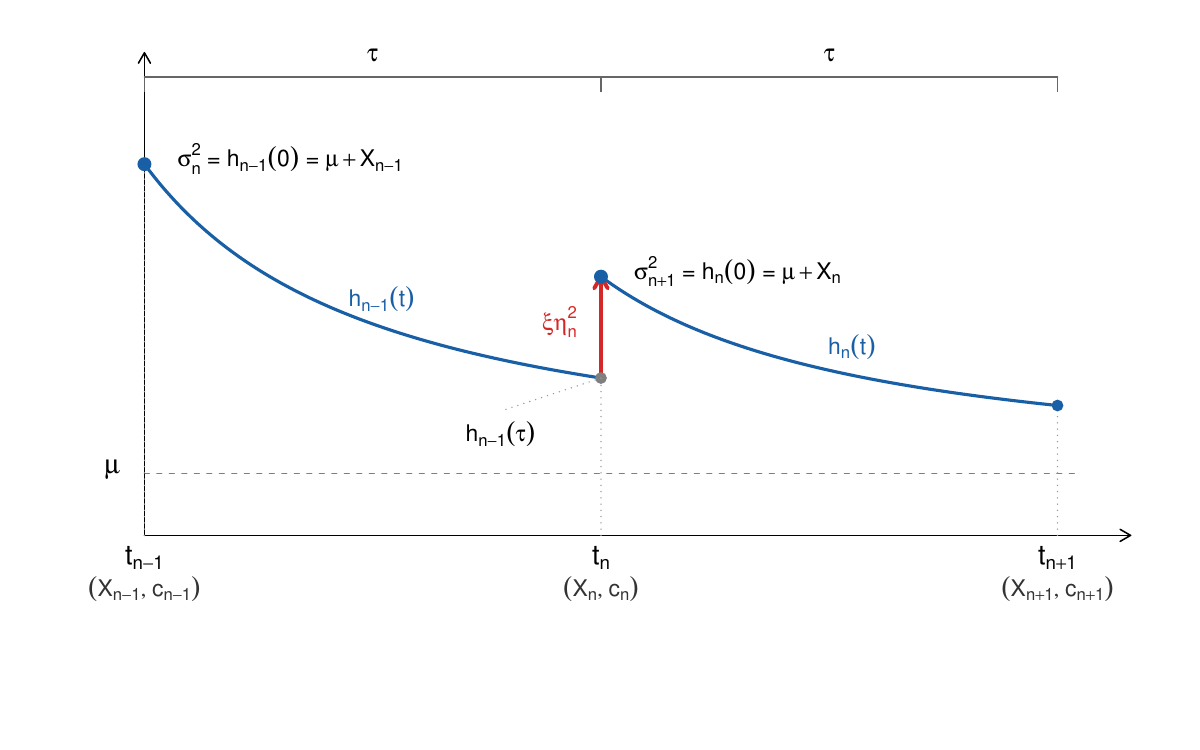}
	\caption{Power-law kernel and the volatility state update}
	\label{fig:kernel-timing}
\end{figure}

Substituting \eqref{Eq:h} into these two matching conditions, and using
\(\partial_t(t+c)^{-p}=-p(t+c)^{-(p+1)}\), yields
\begin{align}
	\frac{a_n}{c_n^p}
	&=
	\frac{a_{n-1}}{(\tau+c_{n-1})^p}
	+
	\frac{\alpha}{\gamma^p}\eta_n^2,
	\label{eq:a_update_level}\\
	\frac{a_n}{c_n^{p+1}}
	&=
	\frac{a_{n-1}}{(\tau+c_{n-1})^{p+1}}
	+
	\frac{\alpha}{\gamma^{p+1}}\eta_n^2.
	\label{eq:c_update_slope}
\end{align}
Setting \(\xi:=\alpha/\gamma^p\), and recalling that
\(X_n=a_n/c_n^p\), the level and slope equations become
\begin{align}
	X_n
	&=
	X_{n-1}\left(1+\frac{\tau}{c_{n-1}}\right)^{-p}
	+
	\xi\eta_n^2,
	\label{eq:X_update}\\
	\frac{X_n}{c_n}
	&=
	\frac{X_{n-1}}{c_{n-1}}
	\left(1+\frac{\tau}{c_{n-1}}\right)^{-(p+1)}
	+
	\frac{\xi}{\gamma}\eta_n^2.
	\label{eq:Y_update}
\end{align}
Together, \eqref{eq:X_update}--\eqref{eq:Y_update} motivate the following
two-dimensional Markov recursion for the updated volatility state
\(\mathbf X_n=(X_n,c_n)\).

We write the innovation in standardized form as
\[
\eta_n=\sigma_n\varepsilon_n,
\qquad
\sigma_n^2=\mu+X_{n-1},
\]
where \((\varepsilon_n)_{n\ge1}\) are i.i.d. with
\(\E [\varepsilon_n] = 0\) and \(\E[\varepsilon_n^2]=1\).  Equivalently, with
\(U_n:=\varepsilon_n^2\),
\begin{equation}
\eta_n^2=(\mu+X_{n-1})U_n, \label{eq:eta_sq}
\end{equation}
where \(U_n>0\) almost surely and \(\E[U_n]=1\).

\begin{definition}[Long-memory GARCH Markov chain]
	\label{Def:Model}
	Let
	\[
	\mathcal S:=[0,\infty)\times[\gamma,\infty)
	\]
	and let
	\[
	\mu>0,\qquad \gamma>0,\qquad \xi>0,\qquad p>1,
	\qquad \tau>0.
	\]
	Given \(\mathbf X_{n-1}=(X_{n-1},c_{n-1})\in\mathcal S\), define
	\[
	\mathbf X_n
	=
	\Psi(\mathbf X_{n-1},\eta_n^2)
	=
	\bigl(
	\Psi_x(X_{n-1},c_{n-1},\eta_n^2),
	\Psi_c(X_{n-1},c_{n-1},\eta_n^2)
	\bigr),
	\]
	where, for \(r(c):=(1+\tau/c)^{-1}\),
	\begin{align}
		\Psi_x(x,c,y)
		&=
		x\,r(c)^p+\xi y,
		\label{Eq:Psi_x}\\
		\Psi_c(x,c,y)
		&=
		\frac{\Psi_x(x,c,y)}
		{\dfrac{x}{c}r(c)^{p+1}
			+\dfrac{\xi}{\gamma}y}.
		\label{Eq:Psi_c}
	\end{align}
\end{definition}

\begin{remark}
	Since \(U_n>0\) almost surely and \(\mu+X_{n-1}>0\), we have
	\(\eta_n^2=(\mu+X_{n-1})U_n>0\) almost surely.  Hence
	\(\Psi_x(x,c,\eta_n^2)>0\) almost surely for all \(x\ge0\), and the
	denominator in \eqref{Eq:Psi_c} is strictly positive almost surely.
	Thus the update map \(\Psi\) is well-defined on \(\mathcal S\).
\end{remark}

The persistence of the model can be seen from its random-coefficient
ARCH representation.  Define the time-varying decay factor
\begin{equation}
	\beta_n
	:=
	r(c_n)^p
	=
	\left(1+\frac{\tau}{c_n}\right)^{-p}
	\in(0,1).
\end{equation}
Then the level recursion \eqref{eq:X_update} can be written as
\begin{equation}
	X_n
	=
	\beta_{n-1}X_{n-1}
	+
	\xi\eta_n^2.
	\label{eq:random-coef-arch}
\end{equation}
By backward iteration, for a stationary version of the process, we obtain
\begin{equation}
	\sigma_{n+1}^2
	=
	\mu
	+
	\xi
	\sum_{k=0}^{\infty}
	\left(\prod_{j=1}^{k}\beta_{n-j}\right)\eta_{n-k}^2,
	\label{eq:arch-infinity-representation}
\end{equation}
where the empty product is interpreted as one.  The weights in this
representation are random and state-dependent, since
\(\beta_n=r(c_n)^p\).  Thus the model has an
\(\mathrm{ARCH}(\infty)\)-type representation with endogenous decay
weights, while the state process remains two-dimensional and Markovian.

\section{Markov chain} \label{sec:Markov}

\subsection{Transition kernel}
\label{subsec:transition-kernel}

We now formulate the process \(\mathbf X_n=(X_n,c_n)\) as a
time-homogeneous Markov chain.  The state space is
\begin{equation}
	\mathcal S
	:=
	[0,\infty)\times[\gamma,\infty),
	\label{eq:state-space}
\end{equation}
where the lower bound \(c_n\ge\gamma\) follows from the update structure.
Using the standardized squared innovation representation 
in Eq.~\eqref{eq:eta_sq},
where \((U_n)_{n\ge1}\) are i.i.d. with common distribution \(F_U\), 
the one-step
transition kernel is
\begin{equation}
	P((x,c),B)
	=
	F_U\left(
	\left\{
	u\ge0:
	\bigl(
	\Psi_x(x,c,(\mu+x)u),
	\Psi_c(x,c,(\mu+x)u)
	\bigr)
	\in B
	\right\}
	\right),
	\label{Eq:P}
\end{equation}
for \(B\in\mathcal B(\mathcal S)\).  Thus \(\mathbf X_n\) is a general
state-space Markov chain on \(\mathcal S\).  The \(n\)-step transition
kernel is denoted by \(P^n\).

When \(F_U\) admits a Lebesgue density \(f_U\), the one-step action of
\(P\) on a nonnegative measurable function \(g:\mathcal S\to\mathbb R_+\) is
\begin{equation}
	Pg(x,c)
	=
	\int_0^\infty
	g\bigl(
	\Psi_x(x,c,(\mu+x)u),
	\Psi_c(x,c,(\mu+x)u)
	\bigr)
	f_U(u)\,\D u.
	\label{eq:Pg}
\end{equation}

\begin{lemma}[One-step bounds for the memory scale]
	\label{lem:c-lower-bound}
	Suppose that \(c_{n-1}\ge\gamma\) and
	\(\eta_n^2>0\) almost surely.  Then
	\begin{equation}
		\gamma
		\le
		c_n
		\le
		c_{n-1}+\tau
		\qquad\text{almost surely}.
		\label{eq:c-one-step-bounds}
	\end{equation}
	Consequently, if \(c_0\ge\gamma\), then
	\(c_n\ge\gamma\) for every \(n\ge0\) almost surely, and
	\[
	c_n-c_{n-1}\le\tau
	\qquad\text{almost surely}.
	\]
\end{lemma}

\begin{proof}
	Fix
	\[
	(X_{n-1},c_{n-1})=(x,c),
	\qquad
	x\ge0,\quad c\ge\gamma,
	\]
	and write \(z=\eta_n^2>0\).  Define
	\[
	A:=x r(c)^p,
	\qquad
	B:=\xi z,
	\]
	where \(r(c)=c/(c+\tau)\).  Since
	\[
	\frac{x}{c}r(c)^{p+1}
	=
	xr(c)^p\frac{r(c)}{c}
	=
	\frac{A}{c+\tau},
	\]
	the update of the scale coordinate can be written as
	\[
	c_n
	=
	\frac{A+B}
	{\dfrac{A}{c+\tau}+\dfrac{B}{\gamma}}.
	\]
	Equivalently,
	\begin{equation}
		\frac1{c_n}
		=
		\frac{A}{A+B}\frac1{c+\tau}
		+
		\frac{B}{A+B}\frac1{\gamma}.
		\label{eq:c-harmonic-combination}
	\end{equation}
	Because \(A\ge0\), \(B>0\), and the two coefficients on the
	right-hand side of \eqref{eq:c-harmonic-combination} are
	nonnegative and sum to one, \(1/c_n\) is a convex combination of
	\(1/(c+\tau)\) and \(1/\gamma\).  Since
	\[
	\gamma\le c+\tau,
	\]
	we have
	\[
	\frac1{c+\tau}
	\le
	\frac1{c_n}
	\le
	\frac1{\gamma}.
	\]
	Taking reciprocals gives
	\[
	\gamma\le c_n\le c+\tau.
	\]
	The remaining conclusions follow immediately.
\end{proof}

The main mathematical objective is to prove that \(\mathbf X_n\) is
positive Harris recurrent, and hence admits a unique stationary
distribution.  
We first prove a Foster--Lyapunov drift condition.  
We then verify the regularity properties needed to apply the Meyn--Tweedie criterion, 
namely \(\psi\)-irreducibility and petite-set regularity.

\subsection{Drift functions and joint stability}
\label{sec:drift}

To establish positive Harris recurrence via the Foster--Lyapunov criterion
(Theorem~\ref{thm:FL} in Appendix~\ref{app:markov-background}), we construct
a Lyapunov function of the form
\[
V(x,c)=\log(1+x)+\lambda c.
\]
Throughout this subsection, \(\mathbb E_{x,c}\) denotes expectation for the
chain started from \((X_0,c_0)=(x,c)\).  We write
\[
\Delta_c(x,c)
:=
\mathbb E_{x,c}[c_1-c],
\qquad
\Delta_X(x,c)
:=
\mathbb E_{x,c}[\log(1+X_1)-\log(1+x)]
\]
for the expected one-step increments of \(c\) and \(\log(1+X)\),
respectively.

\begin{lemma}[Envelope for the \(c\)-drift]
	\label{lem:c-envelope}
	Assume \(U>0\) almost surely. For \((x,c)\in\mathcal S\), let
	\[
	r(c):=\left(1+\frac{\tau}{c}\right)^{-1}.
	\]
	Then
	\begin{equation}
		\Delta_c(x,c)
		=
		-(c-\gamma)
		+
		\frac{c-\gamma r(c)}{r(c)}
		\E\left[
		\frac{\kappa(x,c)}{\kappa(x,c)+U}
		\right],
		\label{eq:drift-exact-rearranged}
	\end{equation}
	where
	\[
	\kappa(x,c)
	:=
	\frac{\gamma r(c)^{p+1}}{\xi c}\frac{x}{\mu+x}.
	\]
	Consequently,
	\[
	\Delta_c(x,c)\le S(c)
	\qquad\text{for all }x\ge0,\ c\ge\gamma,
	\]
	where
	\begin{equation}
		S(c)
		:=
		-(c-\gamma)
		+
		\frac{c-\gamma r(c)}{r(c)}
		\E\left[
		\frac{\kappa^*(c)}{\kappa^*(c)+U}
		\right],
		\qquad
		\kappa^*(c)
		:=
		\frac{\gamma r(c)^{p+1}}{\xi c}.
		\label{eq:S-def}
	\end{equation}
	Moreover, \(S\) is continuous and satisfies
	\(S(\gamma)>0,\)
	\[
	S(c)\le \tau
	\qquad\text{for all }c\ge\gamma,
	\]
	and
	\[
	S(c)\to-\infty
	\qquad\text{as }c\to\infty.
	\]
\end{lemma}

\begin{proof}
	Fix \((X_0,c_0)=(x,c)\), let \(z=\eta_1^2=(\mu+x)U\), and write
	\(r=r(c)\). By the algebraic rearrangement in
	Appendix~\ref{app:c-drift-algebra}, applied to
	\[
	c_1
	=
	\frac{c(xr^p+\xi z)}
	{xr^{p+1}+(\xi c/\gamma)z},
	\]
	we have
	\[
	c_1-c
	=
	-(c-\gamma)
	+
	\frac{c-\gamma r}{r}
	\frac{k}{k+z},
	\qquad
	k:=\frac{\gamma xr^{p+1}}{\xi c}.
	\]
	Taking expectations gives
	\[
	\Delta_c(x,c)
	=
	-(c-\gamma)
	+
	\frac{c-\gamma r}{r}
	\E\left[
	\frac{k}{k+z}
	\right].
	\]
	Because \(z=(\mu+x)U\) and
	\[
	\frac{k}{\mu+x}
	=
	\frac{\gamma r^{p+1}}{\xi c}\frac{x}{\mu+x}
	=
	\kappa(x,c),
	\]
	we obtain
	\[
	\frac{k}{k+z}
	=
	\frac{\kappa(x,c)}{\kappa(x,c)+U}.
	\]
	This proves \eqref{eq:drift-exact-rearranged}.
	
	Next,
	\[
	\kappa(x,c)
	=
	\kappa^*(c)\frac{x}{\mu+x}.
	\]
	For fixed \(c\), the map \(x\mapsto x/(\mu+x)\) is nondecreasing, and
	for each \(U>0\), the map \(y\mapsto y/(y+U)\) is nondecreasing on
	\([0,\infty)\). Hence \(x\mapsto\Delta_c(x,c)\) is nondecreasing.
	Moreover,
	\[
	\kappa(x,c)\nearrow \kappa^*(c)
	\qquad\text{as }x\to\infty.
	\]
	By the monotone convergence theorem,
	\[
	\Delta_c(x,c)
	\longrightarrow
	-(c-\gamma)
	+
	\frac{c-\gamma r(c)}{r(c)}
	\E\left[
	\frac{\kappa^*(c)}{\kappa^*(c)+U}
	\right]
	=
	S(c).
	\]
	Therefore
	\[
	\Delta_c(x,c)\le S(c)
	\qquad\text{for all }x\ge0,\ c\ge\gamma.
	\]
	
	We now prove continuity of \(S\). Let \(c_n\to c\) with
	\(c_n,c\ge\gamma\). Since \(r\) and \(\kappa^*\) are continuous,
	\[
	\kappa^*(c_n)\to \kappa^*(c).
	\]
	Because \(U>0\) almost surely,
	\[
	\frac{\kappa^*(c_n)}{\kappa^*(c_n)+U}
	\longrightarrow
	\frac{\kappa^*(c)}{\kappa^*(c)+U}
	\qquad\text{a.s.}
	\]
	The integrands are bounded by \(1\). Hence, by the dominated convergence
	theorem,
	\[
	\E\left[
	\frac{\kappa^*(c_n)}{\kappa^*(c_n)+U}
	\right]
	\longrightarrow
	\E\left[
	\frac{\kappa^*(c)}{\kappa^*(c)+U}
	\right].
	\]
	It follows from \eqref{eq:S-def} that \(S(c_n)\to S(c)\). Hence \(S\) is
	continuous.
	
	By Lemma~\ref{lem:c-lower-bound},
	\[
	c_1-c\le\tau
	\qquad\text{almost surely}.
	\]
	Therefore
	\[
	\Delta_c(x,c)
	=
	\E_{x,c}[c_1-c]
	\le\tau
	\qquad
	\text{for all }x\ge0,\ c\ge\gamma.
	\]
	Since
	\[
	\Delta_c(x,c)\longrightarrow S(c)
	\qquad\text{as }x\to\infty,
	\]
	it follows that
	\[
	S(c)\le\tau
	\qquad\text{for all }c\ge\gamma.
	\]
	
	Finally, as \(c\to\infty\), \(r(c)\to1\) and
	\[
	\kappa^*(c)
	=
	\frac{\gamma r(c)^{p+1}}{\xi c}
	\longrightarrow0.
	\]
	Since \(U>0\) almost surely,
	\[
	\frac{\kappa^*(c)}{\kappa^*(c)+U}
	\longrightarrow0
	\qquad\text{a.s.}
	\]
	and the integrand is bounded by \(1\). Hence, by dominated convergence,
	\[
	\E\left[
	\frac{\kappa^*(c)}{\kappa^*(c)+U}
	\right]
	\longrightarrow0.
	\]
	Dividing \eqref{eq:S-def} by \(c\), we get
	\[
	\frac{S(c)}{c}
	=
	-\left(1-\frac{\gamma}{c}\right)
	+
	\frac{c-\gamma r(c)}{c\,r(c)}
	\E\left[
	\frac{\kappa^*(c)}{\kappa^*(c)+U}
	\right].
	\]
	Now
	\[
	\frac{c-\gamma r(c)}{c\,r(c)}
	=
	\frac1{r(c)}-\frac{\gamma}{c}
	\longrightarrow1,
	\]
	while the expectation term tends to \(0\). Therefore
	\[
	\lim_{c\to\infty}\frac{S(c)}{c}
	=
	-1.
	\]
	In particular,  \(S(c)\to-\infty\).
	
	At \(c=\gamma\), we have
	\[
	\frac{c-\gamma r(c)}{r(c)}
	=
	\frac{\gamma-\gamma r(\gamma)}{r(\gamma)}
	=
	\tau.
	\]
	Hence
	\[
	S(\gamma)
	=
	\tau
	\E\left[
	\frac{\kappa^*(\gamma)}{\kappa^*(\gamma)+U}
	\right].
	\]
	Since \(\kappa^*(\gamma)>0\) and \(U>0\) almost surely, the expectation is
	strictly positive. Therefore \(S(\gamma)>0\).
\end{proof}

In the Gaussian case, the envelope \(S(c)\) can be written in closed form,
which is useful for numerical verification of the stability condition.

\begin{corollary}[Gaussian formula for the \(c\)-drift envelope]
	\label{cor:gaussian-S}
	Suppose \(U\sim\chi^2(1)\).  For \(\kappa>0\), define
	\[
	h(\kappa)
	:=
	\sqrt{\frac{\pi}{2\kappa}}\,
	\e^{\kappa/2}
	\mathrm{erfc} \left(\sqrt{\kappa/2}\right).
	\]
	Then
	\begin{equation}
		S(c)
		=
		-(c-\gamma)
		+
		\frac{\gamma r(c)^p\{c-\gamma r(c)\}}{\xi c}
		h\left(\kappa^*(c)\right),
		\label{eq:S-gaussian}
	\end{equation}
	where
	\[
	\kappa^*(c)=\frac{\gamma r(c)^{p+1}}{\xi c}.
	\]
	Consequently, \(S(c)<0\) if and only if
	\begin{equation}
		\frac{\gamma r(c)^p\{c-\gamma r(c)\}}{\xi c}
		h\left(\kappa^*(c)\right)
		<
		c-\gamma.
		\label{eq:S-negative-gaussian}
	\end{equation}
\end{corollary}

\begin{proof}
	For \(U\sim\chi^2(1)\),
	\[
	\E\left[
	\frac{1}{\kappa+U}
	\right]
	=
	\sqrt{\frac{\pi}{2\kappa}}\,
	\e^{\kappa/2}
	\mathrm{erfc}\left(\sqrt{\kappa/2}\right)
	=
	h(\kappa).
	\]
	Hence
	\[
	\E\left[
	\frac{\kappa}{\kappa+U}
	\right]
	=
	\kappa h(\kappa).
	\]
	Substituting \(\kappa=\kappa^*(c)\) into \eqref{eq:S-def} yields
	\eqref{eq:S-gaussian}.  The equivalence \eqref{eq:S-negative-gaussian}
	is immediate from \eqref{eq:S-gaussian}.
\end{proof}

Figure~\ref{fig:c-envelope} illustrates the envelope in the Gaussian case
\(U\sim\chi^2(1)\).  Each light curve shows \(\Delta_c(x,\cdot)\) for a
fixed \(x\).  As \(x\to\infty\), the curves increase pointwise to
\(S(c)\).
The curve \(S(c)\) crosses zero, and its first sign-crossing threshold is denoted
by
\begin{equation}
c^*
:=
\inf\{c\ge\gamma:S(c)<0\}. \label{eq:c*}
\end{equation}

\begin{figure}[t]
	\centering
	\includegraphics[width=0.95\linewidth]{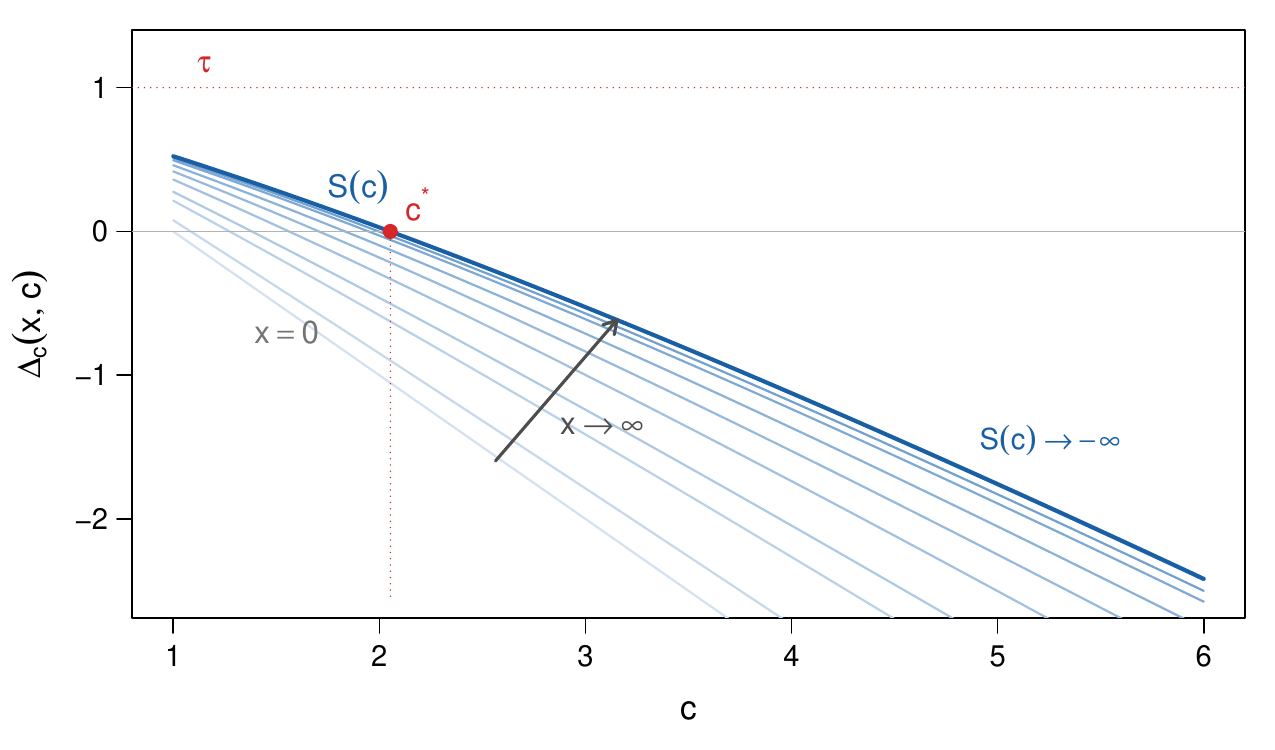}
	\caption{The scale-drift envelope \(S(c)\) for \(U\sim\chi^2(1)\),
		with \(\mu=0.01\), \(p=1.2\), \(\tau=1\), \(\xi=0.5\), \(\gamma=1\).}
	\label{fig:c-envelope}
\end{figure}

Define the large-$x$ log-drift function
\begin{equation}
	\Lambda(c)
	:=
	\E\left[
	\log\left(r(c)^p+\xi U\right)
	\right],
	\qquad c\ge\gamma.
	\label{eq:Lambda-def}
\end{equation}

\begin{lemma}[Properties of the large-\(x\) log-drift]
	\label{lem:Lambda-properties}
	Assume that \(U>0\) almost surely and
	\(\E\log(1+U)<\infty\).  Then the function
	\[
	\Lambda(c)
	=
	\E\left[
	\log\left(r(c)^p+\xi U\right)
	\right],
	\qquad c\ge\gamma,
	\]
	is finite, continuous, and strictly increasing on
	\([\gamma,\infty)\).  Moreover,
	\begin{equation}
		\lim_{c\to\infty}\Lambda(c)
		=
		\E\left[\log(1+\xi U)\right]
		>0.
		\label{eq:Lambda-limit}
	\end{equation}
\end{lemma}

\begin{proof}
	Since \(r(c)\ge r(\gamma)>0\) for \(c\ge\gamma\),
	\[
	\log r(\gamma)^p
	\le
	\log\left(r(c)^p+\xi U\right).
	\]
	Also, because \(r(c)^p\le1\),
	\[
	\log\left(r(c)^p+\xi U\right)
	\le
	\log(1+\xi U)
	\le
	\log\!\bigl(\max\{1,\xi\}\bigr)+\log(1+U).
	\]
	Thus \(\Lambda(c)\) is finite for every \(c\ge\gamma\), and there
	exists a finite constant \(K\) such that
	\[
	\left|
	\log\left(r(c)^p+\xi U\right)
	\right|
	\le
	K+\log(1+U),
	\qquad c\ge\gamma.
	\]
	
	Let \(c_n\to c\).  Since \(r\) is continuous,
	\[
	\log\left(r(c_n)^p+\xi U\right)
	\longrightarrow
	\log\left(r(c)^p+\xi U\right)
	\qquad\text{almost surely}.
	\]
	The preceding integrable bound and the dominated convergence theorem
	therefore imply
	\[
	\Lambda(c_n)\longrightarrow\Lambda(c).
	\]
	Hence \(\Lambda\) is continuous.
	
	Next, if \(\gamma\le c_1<c_2\), then
	\[
	r(c_1)^p<r(c_2)^p.
	\]
	Consequently,
	\[
	\log\left(r(c_1)^p+\xi U\right)
	<
	\log\left(r(c_2)^p+\xi U\right)
	\qquad\text{almost surely},
	\]
	and hence
	\[
	\Lambda(c_1)<\Lambda(c_2).
	\]
	Thus \(\Lambda\) is strictly increasing.
	
	Finally, \(r(c)^p\to1\) as \(c\to\infty\).  Another application of
	the dominated convergence theorem gives
	\[
	\Lambda(c)
	\longrightarrow
	\E\left[\log(1+\xi U)\right].
	\]
	Since \(\xi>0\) and \(U>0\) almost surely,
	\[
	\log(1+\xi U)>0
	\qquad\text{almost surely},
	\]
	so the limiting expectation is strictly positive.
\end{proof}

\begin{lemma}[Uniform large-$x$ bound for the \(\log(1+X)\)-drift]
	\label{lem:X-farfield}
	Assume \(U>0\) almost surely and \(\E\log(1+U)<\infty\).  For every
	\(c_{\max}>\gamma\) and every \(\varepsilon>0\), there exists
	\(M_\varepsilon<\infty\) such that
	\begin{equation}
		\Delta_X(x,c)
		\le
		\Lambda(c)+\varepsilon
		\qquad
		\text{for all }x>M_\varepsilon
		\text{ and all }c\in[\gamma,c_{\max}].
		\label{eq:X-farfield}
	\end{equation}
\end{lemma}

\begin{proof}
	Fix \(c_{\max}>\gamma\).  From \eqref{eq:X_update}, for
	\((X_0,c_0)=(x,c)\),
	\[
	X_1
	=
	xr(c)^p+\xi(\mu+x)U.
	\]
	Hence
	\[
	\frac{1+X_1}{1+x}
	=
	A(x,c)+B(x)U,
	\]
	where
	\[
	A(x,c):=\frac{1+xr(c)^p}{1+x},
	\qquad
	B(x):=\frac{\xi(\mu+x)}{1+x}.
	\]
	Write
	\[
	A(x,c)=r(c)^p+a_x(c),
	\qquad
	B(x)=\xi+b_x,
	\]
	where
	\[
	a_x(c):=\frac{1-r(c)^p}{1+x},
	\qquad
	b_x:=\frac{\xi(\mu-1)}{1+x}.
	\]
	For \(c\in[\gamma,c_{\max}]\),
	\[
	0\le a_x(c)\le\frac{1}{1+x},
	\qquad
	|b_x|\le\frac{\xi|\mu-1|}{1+x}.
	\]
	Thus both error terms vanish uniformly in \(c\in[\gamma,c_{\max}]\).
	
	Since \(r(c)^p\ge r(\gamma)^p>0\) and \(U>0\) almost surely,
	\[
	\frac{a_x(c)+b_xU}{r(c)^p+\xi U}
	\le
	\frac{a_x(c)}{r(c)^p}
	+
	\frac{(b_x)_+U}{\xi U}
	\le
	\frac{1}{(1+x)r(\gamma)^p}
	+
	\frac{(b_x)_+}{\xi}
	=:\theta_x,
	\]
	where \(\theta_x\to0\) as \(x\to\infty\).  Since
	\[
	A(x,c)+B(x)U
	=
	r(c)^p+\xi U+a_x(c)+b_xU
	>0,
	\]
	we may write
	\[
	\log\{A(x,c)+B(x)U\}
	-
	\log\{r(c)^p+\xi U\}
	=
	\log\left(
	1+
	\frac{a_x(c)+b_xU}{r(c)^p+\xi U}
	\right).
	\]
	The fraction inside the logarithm is bounded above by \(\theta_x\), and
	the argument of the logarithm is positive.  Therefore
	\[
	\log\{A(x,c)+B(x)U\}
	-
	\log\{r(c)^p+\xi U\}
	\le
	\log(1+\theta_x).
	\]
	Taking expectations gives
	\[
	\Delta_X(x,c)
	\le
	\Lambda(c)+\log(1+\theta_x),
	\]
	uniformly over \(c\in[\gamma,c_{\max}]\).  Choose \(M_\varepsilon\) such
	that \(\log(1+\theta_x)\le\varepsilon\) for all \(x>M_\varepsilon\).
	This proves \eqref{eq:X-farfield}.
\end{proof}

Figure~\ref{fig:X-farfield} illustrates the large-\(x\) bound in the
Gaussian case \(U\sim\chi^2(1)\).  Each light curve shows
\(\Delta_X(x,\cdot)\) for a fixed \(x\).  For large \(x\), the curves lie
arbitrarily close to the limit bound
\(
\Lambda(c)=\E\log\{r(c)^p+\xi U\},
\)
uniformly on compact \(c\)-intervals.

\begin{figure}[t]
	\centering
	\includegraphics[width=0.95\linewidth]{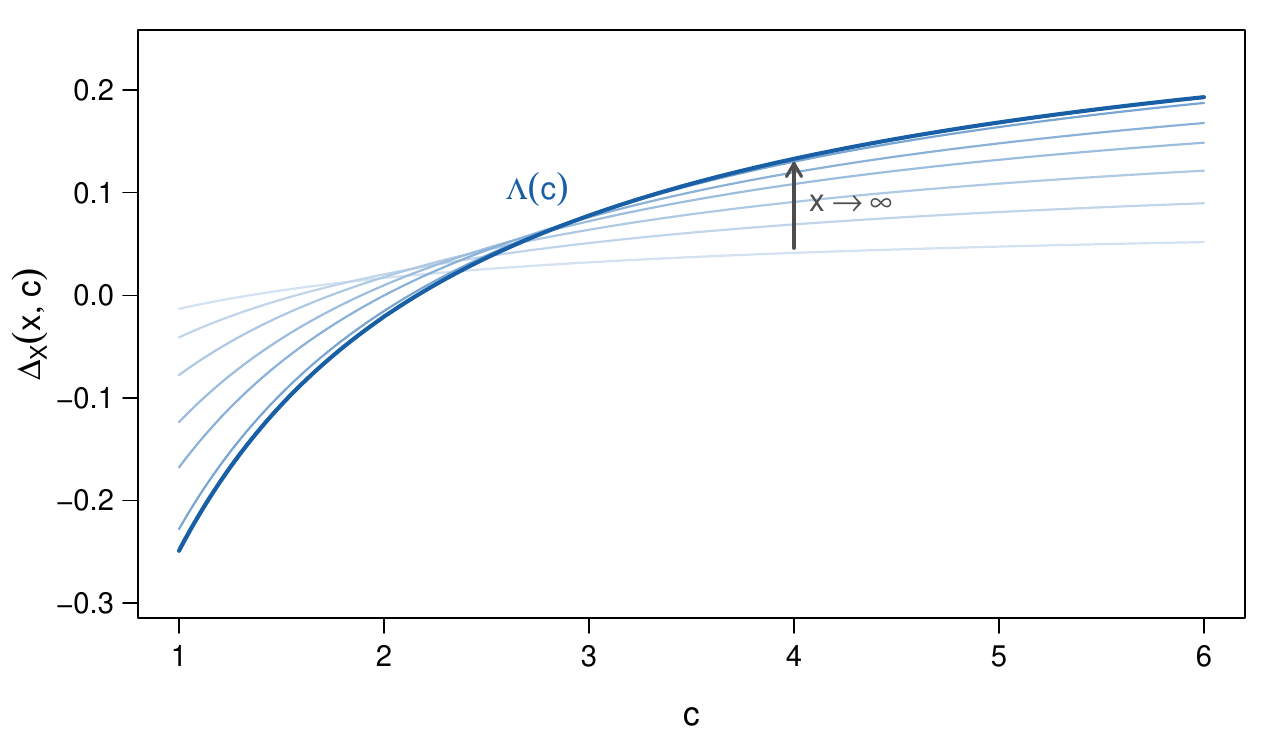}
	\caption{The log-drift \(\Delta_X(x, c)\) and limit $\Lambda(c)$ for \(U\sim\chi^2(1)\),
		with \(\mu=0.01\), \(p=1.2\), \(\tau=1\), \(\xi=0.5\), \(\gamma=1\).}
	\label{fig:X-farfield}
\end{figure}

By Lemmas~\ref{lem:c-envelope} and
\ref{lem:Lambda-properties}, both \(S\) and \(\Lambda\) are continuous
on \([\gamma,\infty)\).  Hence, for every fixed \(\lambda>0\), the joint
large-\(x\) drift function
\[
c\longmapsto \Lambda(c)+\lambda S(c)
\]
is continuous on \([\gamma,\infty)\).

\begin{assumption}[Joint stability]
	\label{ass:joint-stability}
	There exists \(\lambda>0\) such that
	\begin{equation}
		J(\lambda) := \sup_{c\ge\gamma}
		\left\{ \Lambda(c)+\lambda S(c) \right\} <0.
		\label{eq:joint-stability}
	\end{equation}
\end{assumption}

Figure~\ref{fig:joint-stability} illustrates the assumption in the
Gaussian case \(U\sim\chi^2(1)\).  Each curve is
\(c\mapsto\Lambda(c)+\lambda S(c)\) for a fixed \(\lambda\).  As
\(\lambda\) grows, the peak of the curve slides toward smaller \(c\):
a too-small \(\lambda\) leaves the curve above zero at large \(c\), a
too-large \(\lambda\) leaves it above zero at small \(c\), and an
intermediate \(\lambda\) keeps the entire curve below zero.  For such a
\(\lambda\) the supremum over \(c\) is negative, which is exactly
Assumption~\ref{ass:joint-stability}.

\begin{figure}[t]
	\centering
	\includegraphics[width=0.95\linewidth]{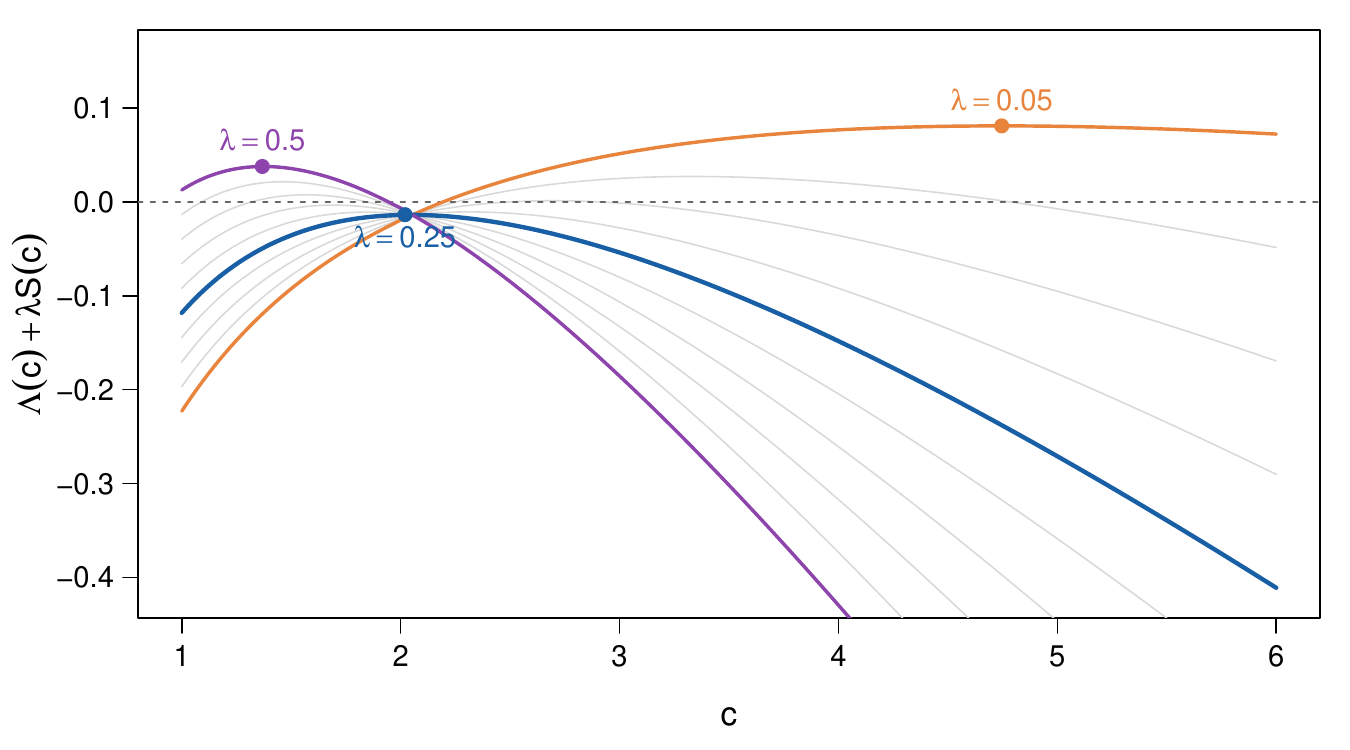}
	\caption{The joint stability quantity \(\Lambda(c)+\lambda S(c)\) for
		several \(\lambda\), with \(U\sim\chi^2(1)\), \(\mu=0.01\),
		\(p=1.2\), \(\tau=1\), \(\xi=0.5\), \(\gamma=1\).}
	\label{fig:joint-stability}
\end{figure}

\begin{proposition}[Joint Foster--Lyapunov condition]
	\label{prop:joint-FL}
	Assume \(U>0\) almost surely, \(\E\log(1+U)<\infty\), and
	Assumption~\ref{ass:joint-stability} with the corresponding
	\(\lambda>0\).  Define
	\begin{equation}
		V(x,c):=\log(1+x)+\lambda c.
		\label{eq:joint-V}
	\end{equation}
	Then there exist \(\delta>0\), \(b<\infty\), \(M<\infty\), and
	\(\bar c>\gamma\) such that
	\begin{equation}
		\E_{x,c}\left[V(X_1,c_1)\right]-V(x,c)
		\le
		-\delta+b\,\mathbf 1_C(x,c),
		\qquad (x,c)\in\mathcal S,
		\label{eq:joint-FL}
	\end{equation}
	where
	\[
	C=[0,M]\times[\gamma,\bar c].
	\]
\end{proposition}

\begin{proof}
	The one-step drift of \(V\) decomposes as
	\begin{equation}
		\E_{x,c}[V(X_1,c_1)]-V(x,c)
		=
		\Delta_X(x,c)+\lambda\Delta_c(x,c).
		\label{eq:V-decomp}
	\end{equation}
	By Lemma~\ref{lem:c-envelope},
	\begin{equation}
		\Delta_X(x,c)+\lambda\Delta_c(x,c)
		\le
		\Delta_X(x,c)+\lambda S(c).
		\label{eq:drift-via-S}
	\end{equation}
	
	First, \(\Delta_X\) is uniformly bounded above.  From
	\eqref{eq:X_update},
	\[
	X_1
	=
	xr(c)^p+\xi(\mu+x)U
	\le
	x+\xi(\mu+x)U.
	\]
	Therefore
	\[
	\frac{1+X_1}{1+x}
	\le
	1+\xi\max(1,\mu)U,
	\]
	and hence
	\begin{equation}
		\Delta_X(x,c)
		\le
		\E\left[
		\log\left(1+\xi\max(1,\mu)U\right)
		\right]
		=:M_D<\infty.
		\label{eq:MD}
	\end{equation}
	
	Let \(\lambda>0\) be the constant from
	Assumption~\ref{ass:joint-stability}, and set
	\[
	\varepsilon:=-\frac{J(\lambda)}{2}>0.
	\]
	Since \(S(c)\to-\infty\), choose \(\bar c>\gamma\) such that
	\begin{equation}
		\lambda S(c)\le -(M_D+1)
		\qquad\text{for all }c>\bar c.
		\label{eq:cbar}
	\end{equation}
	With \(c_{\max}=\bar c\), Lemma~\ref{lem:X-farfield} gives
	\(M:=M_\varepsilon<\infty\) such that
	\begin{equation}
		\Delta_X(x,c)\le\Lambda(c)+\varepsilon
		\qquad
		\text{for all }x>M
		\text{ and all }c\in[\gamma,\bar c].
		\label{eq:X-large-region}
	\end{equation}
	
	We now consider three regions.  If \(c>\bar c\), then by
	\eqref{eq:drift-via-S}, \eqref{eq:MD}, and \eqref{eq:cbar},
	\[
	\Delta_X(x,c)+\lambda\Delta_c(x,c)
	\le
	M_D+\lambda S(c)
	\le
	-1.
	\]
	If \(c\le \bar c\) and \(x>M\), then by
	\eqref{eq:drift-via-S}, \eqref{eq:X-large-region}, and
	Assumption~\ref{ass:joint-stability},
	\[
	\Delta_X(x,c)+\lambda\Delta_c(x,c)
	\le
	\Lambda(c)+\varepsilon+\lambda S(c)
	\le
	J(\lambda) + \varepsilon =
	\frac{J(\lambda)}{2} <0.
	\]
	Finally, if \(c\le \bar c\) and \(x\le M\), then
	\eqref{eq:MD} and \(S(c)\le\tau\) give
	\[
	\Delta_X(x,c)+\lambda\Delta_c(x,c) \le
	M_D+\lambda\tau =:b_0.
	\]
	The set
	\[
	C=[0,M]\times[\gamma,\bar c]
	\]
	is compact in \(\mathcal S\).  Set
	\[
	\delta
	:=
	\min\left\{1,-\frac{J(\lambda)}{2}\right\}>0,
	\qquad
	b:=b_0+\delta.
	\]
	Outside \(C\), either \(c>\bar c\) or \(x>M\) with \(c\le \bar c\), so the
	drift is bounded above by \(-\delta\).  On \(C\), the drift is bounded
	above by \(b_0=-\delta+b\).  Therefore
	\[
	\E_{x,c}[V(X_1,c_1)]-V(x,c)
	\le
	-\delta+b\,\mathbf 1_C(x,c),
	\]
	for all \((x,c)\in\mathcal S\).
\end{proof}

\subsection{Diagnostic condition}
\label{sec:joint-stability-interpretation}
	
Assumption~\ref{ass:joint-stability} is the large-$x$ drift condition
for the Lyapunov function
\[
V(x,c)=\log(1+x)+\lambda c.
\]
As \(x\to\infty\), the corresponding drift is controlled by
\[
\Lambda(c)+\lambda S(c),
\qquad
\Lambda(c)
=
\E\log\{r(c)^p+\xi U\},
\]
where \(S(c)\) is the upper envelope for the large-$x$ drift of the
scale coordinate.  Thus the joint stability condition requires a
single \(\lambda>0\) such that
\[
\sup_{c\ge\gamma}\{\Lambda(c)+\lambda S(c)\}<0.
\]

By Lemma~\ref{lem:c-envelope}, \(S\) is continuous,
\(S(\gamma)>0\), and \(S(c)\to-\infty\) as \(c\to\infty\).
Hence by \eqref{eq:c*}, \(c^*<\infty\), \(S(c^*)=0\), and
\[
S(c)\ge0
\qquad\text{for all }c\in[\gamma,c^*].
\]

This threshold gives a useful necessary diagnostic for the joint stability
condition.  Indeed, if Assumption~\ref{ass:joint-stability} holds, then
for every \(c\in[\gamma,c^*]\),
\[
\Lambda(c)+\lambda S(c)<0.
\]
Since \(S(c)\ge0\) on this interval, it follows that
\[
\Lambda(c)<0
\qquad\text{for all }c\in[\gamma,c^*].
\]
By Lemma~\ref{lem:Lambda-properties}, \(\Lambda\) is strictly increasing
on \([\gamma,\infty)\).  Consequently, the strongest necessary
diagnostic on \([\gamma,c^*]\) is obtained at \(c=c^*\). 
Thus Assumption~\ref{ass:joint-stability}
implies
\begin{equation}
	\Lambda(c^*)
	=
	\E\left[\log(\rho^*+\xi U)\right]
	<0,
	\qquad
	\rho^*
	:=
	\left(\frac{c^*}{c^*+\tau}\right)^p.
	\label{eq:diagnostic-log-condition}
\end{equation}

Condition~\eqref{eq:diagnostic-log-condition} is therefore a necessary
condition for Assumption~\ref{ass:joint-stability}.  It is the analogue
of the usual GARCH log-moment condition
\[
\E[\log(\beta+\alpha U)]<0,
\]
with the fixed persistence coefficient replaced by the endogenous
coefficient \(\rho^*\) evaluated at the critical scale \(c^*\).  We do not
claim that this diagnostic condition is sufficient for
Assumption~\ref{ass:joint-stability}; sufficiency would require uniform
negativity of
\[
\Lambda(c)+\lambda S(c)
\]
over all \(c\ge\gamma\) for a single \(\lambda>0\).  Nevertheless, in the
Gaussian simulations below, this diagnostic boundary is numerically very
close to the boundary obtained by directly checking the joint stability
condition.

Figure~\ref{fig:stationarity} shows the zero-level boundary of
\eqref{eq:diagnostic-log-condition} computed on a
\((\gamma,\xi)\)-grid for \(p=1.2\), \(\tau=1\), and
\(U\sim\chi^2(1)\).
For comparison, we also computed the zero boundary of the joint stability
value
\[
J_{\min}(\gamma,\xi) :=
\inf_{\lambda>0}\sup_{c\ge\gamma}
\{\Lambda(c)+\lambda S(c)\}
\]
on the same grid.  The diagnostic and joint criteria give identical
stable/unstable classifications at all \(8100\) grid points.  Their
boundary locations differ by at most \(3.6\times10^{-4}\) when compared
as \(\xi=\xi(\gamma)\), and by at most \(7.3\times10^{-5}\) when compared
as \(\gamma=\gamma(\xi)\).

\begin{figure}[t]
\centering
\includegraphics[width=0.9\linewidth]{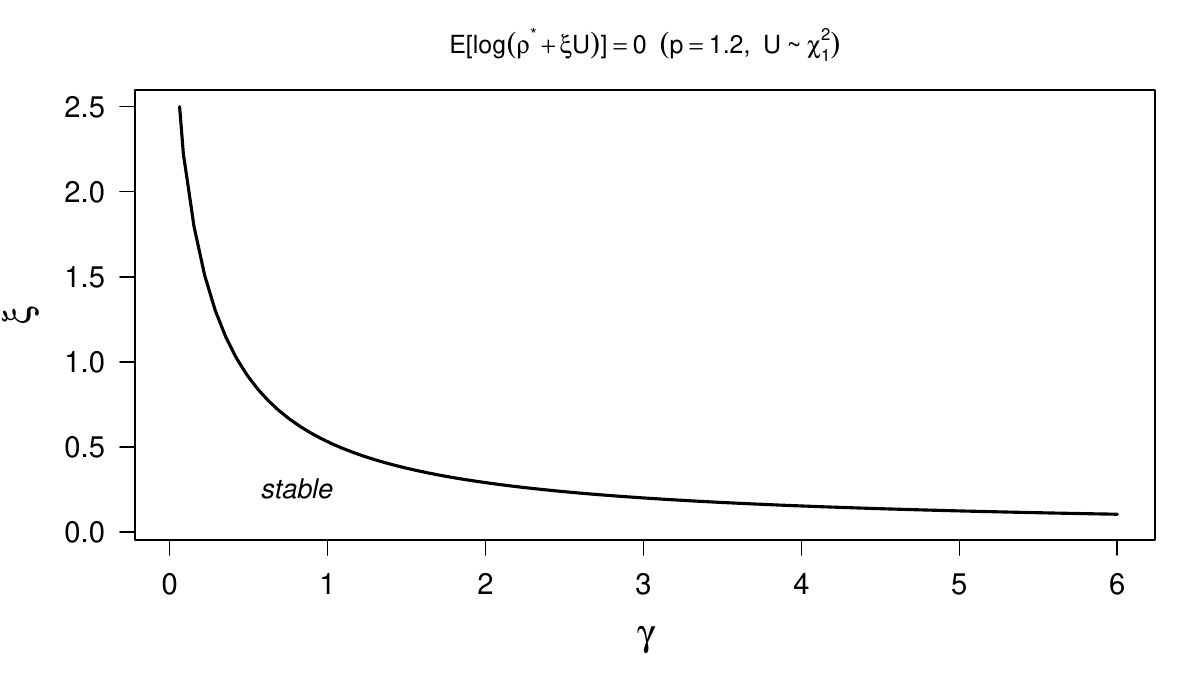}
\caption{Diagnostic stability boundary in the \((\gamma,\xi)\)-plane
	for \(p=1.2\), \(\tau=1\), and \(U\sim\chi^2(1)\).}
\label{fig:stationarity}
\end{figure}

\section{Control model and regularity}
\label{sec:regularity}

We introduce the deterministic control model associated with the
transition kernel of the Markov chain.  For notational clarity, define
the control map
\[
\Phi((x,c),u)
:=
\bigl(
\Psi_x(x,c,(\mu+x)u),
\Psi_c(x,c,(\mu+x)u)
\bigr),
\qquad (x,c)\in\mathcal S,\ u>0.
\]
Thus \(u\) plays the role of the standardized squared innovation \(U_n\),
whereas \((\mu+x)u\) is the corresponding squared innovation
\(\eta_n^2\).

Let
\[
O_U:=\{u>0:f_U(u)>0\}.
\]
When \(f_U\) is strictly positive on \((0,\infty)\), we have
\(O_U=(0,\infty)\).
The associated deterministic control model is
\[
\mathbf X_n=\Phi(\mathbf X_{n-1},u_n),
\qquad u_n\in O_U,\quad n\ge1.
\]
Equivalently, if \(\mathbf X_n=(X_n,c_n)\), then
\begin{align}
	X_n
	&=
	X_{n-1}r(c_{n-1})^p
	+
	\xi(\mu+X_{n-1})u_n,
	\label{eq:cm-X}\\
	\frac{X_n}{c_n}
	&=
	\frac{X_{n-1}}{c_{n-1}}r(c_{n-1})^{p+1}
	+
	\frac{\xi}{\gamma}(\mu+X_{n-1})u_n.
	\label{eq:cm-Y}
\end{align}

For \(\mathbf x\in\mathcal S\), define its forward control orbit by
\[
A_+(\mathbf x)
:=
\{\mathbf x\}
\cup
\left\{
\Phi^{(N)}(\mathbf x;u_1,\dots,u_N):
N\ge1,\ (u_1,\dots,u_N)\in O_U^N
\right\},
\]
where \(\Phi^{(N)}\) denotes the \(N\)-fold composition of the control map.

The regularity argument below follows the standard control-model
approach to irreducibility and petiteness for nonlinear state-space
Markov chains; see \citet[Chapter~7]{meyn2009markov} and
\citet{ChotardAuger2019}.  The next lemma shows that, in the present
model, the control model has the boundary point \((0,\gamma)\) as a
globally attracting anchor.  The subsequent local minorization is then a
boundary-anchor analogue of the usual rank-condition-to-small-set
argument.

\begin{lemma}[Attraction to the boundary anchor]
	\label{lem:attractor}
	Assume that \(0\in\overline{O_U\cap(0,\infty)}\), that is, positive
	admissible controls can be chosen arbitrarily close to zero.  Then, for
	every \((x_0,c_0)\in\mathcal S\),
	\[
	(0,\gamma)\in\overline{A_+((x_0,c_0))}.
	\]
	Equivalently, for every relatively open neighborhood \(N\) of
	\((0,\gamma)\) in \(\mathcal S\), there exist \(n\ge1\) and controls
	\((u_1,\ldots,u_n)\in O_U^n\) such that
	\[
	\Phi^{(n)}((x_0,c_0);u_1,\ldots,u_n)\in N.
	\]
\end{lemma}

\begin{proof}
	See Appendix~\ref{app:attractor-proof}.
\end{proof}

For \(z=(x_0,c_0)\in\mathcal S\) and positive controls, we denote by
\[
\mathbf J_3(z;a_1;u_2,u_3)
\quad\text{and}\quad
\mathbf J_4(z;a_1,a_2;u_3,u_4)
\]
the \(2\times2\) Jacobian matrices, respectively, of the maps
\[
(u_2,u_3)\mapsto \Phi^{(3)}(z;a_1;u_2,u_3)
\quad\text{and}\quad
(u_3,u_4)\mapsto \Phi^{(4)}(z;a_1,a_2;u_3,u_4).
\]
Thus, in each case, the Jacobian is taken with respect to the last two
control variables, while the preceding controls and the initial state are
held fixed.

For \(i\le j\), write \(a_{i:j}:=(a_i,\ldots,a_j)\) and
let $\mathcal S^\circ$ be the interior of $\mathcal S$, 
i.e., $\mathcal S^\circ=(0,\infty)\times(\gamma,\infty).$

\begin{lemma}[Anchor rank condition]
	\label{lem:anchor-rank}
	At least one of the following alternatives holds.
	\begin{enumerate}
		\item There exist controls \(a_{1:3}\in(0,\infty)^3\) such that
		\[
		\Phi^{(3)}((0,\gamma);a_1;a_2,a_3)
		\in\mathcal S^\circ
		\quad\text{and}\quad
		\det \mathbf J_3((0,\gamma);a_1;a_2,a_3)\ne0.
		\]
		
		\item There exist controls \(a_{1:4}\in(0,\infty)^4\) such that
		\[
		\Phi^{(4)}((0,\gamma);a_1,a_2;a_3,a_4)
		\in\mathcal S^\circ
		\quad\text{and}\quad
		\det \mathbf J_4((0,\gamma);a_1,a_2;a_3,a_4)\ne0.
		\]
	\end{enumerate}
\end{lemma}

\begin{proof}
	See Appendix~\ref{app:det-rank}.
\end{proof}

The next lemma turns the anchor rank condition into a local minorization.
The full-rank dependence on the last two controls creates a two-dimensional
density on an interior ball, while lower semicontinuity and strict
positivity of \(f_U\) give a uniform positive lower bound on compact
control sets.

\begin{lemma}[Anchor minorization]
	\label{lem:anchor-minorization}
	Assume that \(f_U\) is lower semicontinuous and strictly positive on
	\((0,\infty)\).  Then there exist a relatively open neighborhood
	\(N_0\) of \((0,\gamma)\) in \(\mathcal S\), an integer
	\(m\in\{3,4\}\), a nonempty open ball
	\[
	B\Subset\mathcal S^\circ
	\]
	and a constant \(\delta>0\) such that
	\begin{equation}
		P^m(z,A)
		\ge
		\delta\,\operatorname{Leb}(A\cap B),
		\qquad
		z\in N_0,\quad A\in\mathcal B(\mathcal S).
		\label{eq:anchor-minorization}
	\end{equation}
\end{lemma}
\begin{proof}
	By Lemma~\ref{lem:anchor-rank}, choose \(m\in\{3,4\}\) and positive
	controls \(a_{1:m}\in(0,\infty)^m\) such that
	\[
	\Phi^{(m)}((0,\gamma);a_{1:m})\in\mathcal S^\circ 
	\quad \text{and} \quad
	\det \mathbf J_m((0,\gamma);a_{1:m-2};a_{m-1:m})\ne0,
	\]
	where \(\mathbf J_m\) denotes \(\mathbf J_3\) when \(m=3\) and
	\(\mathbf J_4\) when \(m=4\).
	
	For \(v=(u_1,\ldots,u_{m-2})\) and
	\(q=(u_{m-1},u_m)\), write
	\[
	\Phi^{(m)}(z;v,q)
	:=
	\Phi^{(m)}(z;u_1,\ldots,u_m).
	\]
	By the parameterized inverse function theorem and continuity of the
	Jacobian matrix with respect to the initial state and the control
	variables, there exist a relatively open neighborhood \(N_0\) of
	\((0,\gamma)\), compact control neighborhoods
	\[
	V_0\Subset (0,\infty)^{m-2},
	\qquad
	V_1\Subset (0,\infty)^2,
	\]
	containing \(a_{1:m-2}\) and \(a_{m-1:m}\) in their respective interiors,
	and a nonempty open ball
	\[
	B\Subset\mathcal S^\circ
	\]
	such that, for every \(z\in N_0\) and every \(v\in V_0\), the map
	\[
	q\mapsto \Phi^{(m)}(z;v,q),
	\qquad q\in V_1,
	\]
	is one-to-one on \(V_1\) and its image contains the same ball \(B\).
	We may take \(V_0\) and \(V_1\) to be the closures of open rectangles;
	in particular,
	\[
	\operatorname{Leb}_{m-2}(V_0)>0.
	\]
	Moreover, by continuity of the Jacobian determinant and compactness of the relevant
	sets, there exists \(J_0<\infty\) such that
	\[
	|\det \mathbf J_m(z;v;q)|\le J_0
	\]
	on the relevant preimage of \(B\).
	
	Since \(V_0\times V_1\) is a compact subset of
	\((0,\infty)^m\), all coordinate projections of \(V_0\times V_1\) are compact
	subsets of \((0,\infty)\).
	Hence, by lower semicontinuity and strict positivity of \(f_U\), there exists
	\(\alpha>0\) such that
	\[
	f_U(u_j)\ge\alpha
	\qquad
	\text{for all }(u_1,\ldots,u_m)\in V_0\times V_1
	\text{ and }j=1,\ldots,m.
	\]
	
	For fixed \(z\in N_0\) and \(v\in V_0\), since
	\(B\subset \Phi^{(m)}(z;v,V_1)\) and the map
	\(q\mapsto \Phi^{(m)}(z;v,q)\) is one-to-one on \(V_1\), the
	change-of-variables formula for \(y=\Phi^{(m)}(z;v,q)\) gives
	\[
	\begin{aligned}
		\int_{V_1}
		\mathbf 1_A\left(\Phi^{(m)}(z;v,q)\right)\,\D q
		\ge
		\int_{A\cap B}
		\frac{1}
		{\left|\det \mathbf J_m(z;v;q(y))\right|}
		\,\D y  \ge
		J_0^{-1}\operatorname{Leb}(A\cap B),
	\end{aligned}
	\]
	where \(q(y)\in V_1\) is the unique point satisfying
	\[
	\Phi^{(m)}(z;v,q(y))=y.
	\]
	Thus,
	\[
	\begin{aligned}
		P^m(z,A)
		&\ge
		\int_{V_0}
		\int_{V_1}
		\mathbf 1_A\left(\Phi^{(m)}(z;v,q)\right)
		\prod_{j=1}^m f_U(u_j) \,\D q\,\D v  \\
		&\ge
		\alpha^m \int_{V_0}\int_{V_1}
		\mathbf 1_A(\Phi^{(m)}(z;v,q)) \,\D q \,\D v \\
		&\ge
		\alpha^m J_0^{-1}\operatorname{Leb}_{m-2}(V_0)
		\operatorname{Leb}(A\cap B).
	\end{aligned}
	\]
	Therefore
	\eqref{eq:anchor-minorization} holds with
	\[
	\delta := \alpha^m J_0^{-1}\operatorname{Leb}_{m-2}(V_0)>0.
	\]
\end{proof}

\begin{lemma}[Feller property]
	\label{lem:feller}
	Assume that \(U\) has a Lebesgue density \(f_U\) on \((0,\infty)\).
	Then the Markov chain \((\mathbf X_n)_{n\ge0}\) is Feller on
	\(\mathcal S\).
\end{lemma}
\begin{proof}
	Let \(g\in C_b(\mathcal S)\).  For \((x,c)\in\mathcal S\),
	\[
	Pg(x,c)
	=
	\int_0^\infty
	g\left(\Phi((x,c),u)\right) f_U(u)\,\D u.
	\]
	For each \(u>0\), the map
	\[
	(x,c)\mapsto g\left(\Phi((x,c),u)\right)
	\]
	is continuous on \(\mathcal S\), because the denominator 
	in the second component of \(\Phi\) is strictly positive.  
	Moreover, it is bounded
	by \(\|g\|_\infty\).  Therefore the dominated convergence theorem gives
	the continuity of \(Pg\).  Hence the chain is Feller.
\end{proof}

The next theorem assembles the regularity argument.  Attraction to the
boundary anchor gives pointwise access to the minorizing neighborhood,
and the local minorization on \(B\) yields irreducibility.  The Feller
property is used to make this access locally uniform, so that a finite
covering argument gives petiteness of compact sets.

\begin{theorem}[Regularity]
	\label{thm:regularity}
	Assume that \(f_U\) is lower semicontinuous and strictly positive on
	\((0,\infty)\).  Then the Markov chain \((\mathbf X_n)_{n\ge0}\) on
	\(\mathcal S\) is \(\phi\)-irreducible with
	\[
	\phi(A):=\operatorname{Leb}(A\cap B),
	\qquad
	A\in\mathcal B(\mathcal S),
	\]
	where \(B\) is as in Lemma~\ref{lem:anchor-minorization}.  In
	particular, the chain is \(\psi\)-irreducible.  Moreover, every compact
	subset of \(\mathcal S\) is petite.
\end{theorem}

\begin{proof}
	Let \(N_0\), \(m\), \(B\), and \(\delta\) be as in
	Lemma~\ref{lem:anchor-minorization}, and let \(\phi\) be the measure
	defined in the statement.
	
	Since \(f_U\) is strictly positive on \((0,\infty)\), we have
	\(O_U=(0,\infty)\), and hence \(0\in\overline{O_U\cap(0,\infty)}\).
	Thus Lemma~\ref{lem:attractor} applies.
	
	We first prove \(\phi\)-irreducibility.  Let
	\(\mathbf z\in\mathcal S\).  By Lemma~\ref{lem:attractor}, for the
	neighborhood \(N_0\) of \((0,\gamma)\), there exists a finite
	admissible control path whose endpoint lies in \(N_0\).  Since the
	finite-step control map is continuous in the controls, the inverse image
	of \(N_0\) contains a nonempty open subset of the control space.  Since
	\(f_U\) is strictly positive on \((0,\infty)\), this open set has
	positive probability.  Hence
	\[
	P^n(\mathbf z,N_0)>0
	\]
	for some \(n\ge0\).  Combining this with
	Lemma~\ref{lem:anchor-minorization}, we obtain
	\[
	P^{n+m}(\mathbf z,A)
	\ge
	\delta P^n(\mathbf z,N_0)\operatorname{Leb}(A\cap B).
	\]
	Thus, if \(\phi(A)>0\), then \(P^{n+m}(\mathbf z,A)>0\).  Hence the
	chain is \(\phi\)-irreducible.  Therefore, by the existence of a maximal
	irreducibility measure, the chain is \(\psi\)-irreducible.
	
	It remains to show that compact sets are petite.  
	Let \(K\subset\mathcal S\) be compact.  
	For each \(\mathbf z\in K\), the reachability
	argument above gives an integer \(n_{\mathbf z}\) such that
	\[
	P^{n_{\mathbf z}}(\mathbf z,N_0)>0.
	\]
	By Lemma~\ref{lem:feller}, the chain is Feller, and hence \(P^n\) is
	Feller for every \(n\ge1\).
	Since \(N_0\) is relatively open in \(\mathcal S\), the map
	\[
	\mathbf y\mapsto P^{n_{\mathbf z}}(\mathbf y,N_0)
	\]
	is lower semicontinuous.  Hence there exist a neighborhood
	\(U_{\mathbf z}\) of \(\mathbf z\) and a constant \(r_{\mathbf z}>0\)
	such that
	\[
	P^{n_{\mathbf z}}(\mathbf y,N_0)\ge r_{\mathbf z},
	\qquad \mathbf y\in U_{\mathbf z}.
	\]
	Choose a finite subcover
	\[
	K\subset\bigcup_{j=1}^J U_{\mathbf z_j}.
	\]
	Set
	\[
	n_j:=n_{\mathbf z_j},
	\qquad
	r_j:=r_{\mathbf z_j}.
	\]
	Choose a sampling distribution \(a\) on \(\mathbb Z_+\) such that
	\[
	a(n_j+m)>0,
	\qquad j=1,\ldots,J.
	\]
	Then, for \(\mathbf y\in K\), choose \(j\) such that
	\(\mathbf y\in U_{\mathbf z_j}\).  We have
	\[
	\begin{aligned}
		K_a(\mathbf y,A)
		&:=
		\sum_{\ell\ge0}a(\ell)P^\ell(\mathbf y,A) \\
		&\ge
		a(n_j+m)P^{n_j+m}(\mathbf y,A) \\
		&\ge
		a(n_j+m)r_j\delta\,\operatorname{Leb}(A\cap B).
	\end{aligned}
	\]
	Therefore
	\[
	K_a(\mathbf y,A)
	\ge
	\varepsilon\,\phi(A),
	\qquad
	\mathbf y\in K,
	\]
	where
	\[
	\varepsilon
	:=
	\delta
	\min_{1\le j\le J}a(n_j+m)r_j
	>0.
	\]
	Thus \(K\) is petite.
\end{proof}

We can now state the main stability result of the paper.

\begin{theorem}[Positive Harris recurrence of \(\mathbf X_n\)]
	\label{thm:positive-harris}
	Assume the following conditions.
	\begin{enumerate}
		\item[\upshape(A1)] The normalized innovation \(U\) has a Lebesgue
		density \(f_U\) that is lower semicontinuous and strictly positive on
		\((0,\infty)\), satisfies \(\mathbb P(U>0)=1\), and is normalized by
		\(\mathbb E[U]=1\).
		
		\item[\upshape(A2)] Assumption~\ref{ass:joint-stability} holds.
	\end{enumerate}
	Then \(\mathbf X_n=(X_n,c_n)\) is positive Harris recurrent. In
	particular, it admits a unique invariant probability measure on
	\(\mathcal S\).
\end{theorem}

\begin{proof}
	By Theorem~\ref{thm:regularity}, \(\mathbf X_n\) is
	\(\psi\)-irreducible, and every compact subset of \(\mathcal S\) is
	petite.
	
	Let \(\lambda>0\) be chosen so that (A2) holds.
	Since \(\log(1+t)\le t\) for \(t\ge0\), and since \(U\ge0\) and
	\(\mathbb E[U]=1\), we have
	\[
	\mathbb E \left[ \log(1+U) \right] \le \mathbb E[U]=1<\infty.
	\]
	Thus the hypotheses of Proposition~\ref{prop:joint-FL} are satisfied.
	Hence there exist \(\delta>0\), \(b<\infty\), \(M<\infty\), and
	\(\bar c>\gamma\) such that
	\[
	\mathbb E_{x,c}[V(X_1,c_1)]-V(x,c)
	\le
	-\delta + b\,\mathbf 1_C(x,c),
	\qquad (x,c)\in\mathcal S,
	\]
	where
	\[
	V(x,c)=\log(1+x)+\lambda c
	\]
	and
	\[
	C=[0,M]\times[\gamma,\bar c].
	\]
	Adding a constant to \(V\), if necessary, we may assume \(V\ge1\)
	without changing the drift inequality.
	
	The set \(C\) is compact in \(\mathcal S\), and therefore petite by
	Theorem~\ref{thm:regularity}.  The Foster--Lyapunov criterion,
	Theorem~\ref{thm:FL}, then implies that \(\mathbf X_n\) is positive Harris
	recurrent.  The existence and uniqueness of the invariant probability
	measure follow from positive Harris recurrence and
	\(\psi\)-irreducibility.
\end{proof}

\section{Long-memory behavior}
\label{sec:long-memory}

The preceding sections establish stability of the two-dimensional Markov
state process under the joint stability condition.  This section studies
the persistence generated by the model.

The persistence mechanism comes from the power-law kernel together with
the level-and-slope matching update.  Through slope matching, each shock
affects both the current variance level and the endogenous future decay
scale.  Thus the observable volatility proxy may display
long-memory-type behavior even though the augmented state remains
Markovian.

We measure this persistence in simulations using the local Whittle
estimator applied to log-squared innovations.

\subsection{Local Whittle estimation}
\label{sec:lw}

Let \(\{Y_t\}\) be a covariance-stationary process with spectral density
\(f_Y\).  We say that \(Y_t\) has fractional memory parameter
\(d\in(-1/2,1/2)\) if
\begin{equation}
	f_Y(\lambda)
	\sim
	C_y \lambda^{-2d},
	\qquad
	\lambda\to0^+,
	\label{eq:spec-pole}
\end{equation}
for some \(C_y>0\).  The range \(d\in(0,1/2)\) corresponds to long memory.

We measure low-frequency persistence using the local Whittle estimator of
\citet{robinson1995gaussian}.  For a demeaned sample
\(Y_1,\ldots,Y_n\), define the periodogram
\begin{equation}
	I(\lambda_j)
	=
	\frac{1}{2\pi n}
	\left|
	\sum_{t=1}^{n}Y_t \e^{-it\lambda_j}
	\right|^2,
	\qquad
	\lambda_j=\frac{2\pi j}{n}.
	\label{eq:periodogram}
\end{equation}
For a bandwidth \(m\), the local Whittle objective is
\begin{equation}
	Q(d)
	=
	\log\left(
	\frac1m
	\sum_{j=1}^{m}\lambda_j^{2d}I(\lambda_j)
	\right)
	-
	\frac{2d}{m}
	\sum_{j=1}^{m}\log\lambda_j,
	\label{eq:lw-obj}
\end{equation}
and
\[
\hat d
=
\argmin_{d\in\mathcal D} Q(d),
\qquad
\mathcal D=[-0.4,0.99].
\]
Under standard regularity conditions, if \(d_0\in(-1/2,1/2)\),
\(m\to\infty\), and \(m/n\to0\), then
\begin{equation}
	\sqrt m(\hat d-d_0)
	\Rightarrow
	N(0,1/4).
	\label{eq:lw-clt}
\end{equation}

In the simulations below, we apply the estimator to demeaned
log-squared innovations.  Values \(0<\hat d<1/2\) are interpreted in the
usual covariance-stationary long-memory range.  Values \(\hat d\ge1/2\)
are reported only as indicators of very strong low-frequency persistence,
not as estimates covered by the stationary local Whittle limit theory.

\subsection{Simulation design}
\label{sec:sim-design}

We apply the local Whittle estimator to demeaned log-squared innovations
\(\{\log \eta_n^2\}\), treating them as observable noisy proxies for
latent log-volatility.  This practice is standard in the long-memory
volatility literature; see, for example, \citet{Ding1993} and
\citet{Breidt1998}.

We set \(m=\lfloor n^{0.65}\rfloor\).  With \(n=20{,}000\), this gives
\(m=624\).  For each \((p,\xi)\) in the grid
\[
\{1.1,1.2,1.5,2.0\}
\times
\{0.10,0.25,0.50,0.75,1.00,1.10\},
\]
we compute the diagnostic boundary value \(\gamma^*\) defined by
\begin{equation}
	\ell(\gamma^*)=0,
	\qquad
	\ell(\gamma)
	:=
	\E\left[
	\log\{\rho^*(\gamma)+\xi U\}
	\right],
	\qquad
	U\sim\chi^2(1),
	\label{eq:lyap-zero}
\end{equation}
where
\begin{equation}
	\rho^*(\gamma)
	:=
	\left(
	\frac{c^*(\gamma)}
	{c^*(\gamma)+\tau}
	\right)^p.
	\label{eq:rho-gamma}
\end{equation}

Here \(c^*(\gamma)\) is the first zero-crossing threshold of the
scale-drift envelope \(S(c)\), with its dependence on \(\gamma\) made
explicit.  That is,
\[
c^*(\gamma):=\inf\{c\ge\gamma:S_\gamma(c)<0\},
\]
where \(S_\gamma(c)\) denotes the function \(S(c)\) from
Section~\ref{sec:joint-stability-interpretation} when \(\gamma\) is
treated as a varying parameter.  The envelope is evaluated using the
Gaussian closed-form expression in Corollary~\ref{cor:gaussian-S}.
Thus \(\gamma^*\) defines a diagnostic boundary for the simulations,
rather than the formal joint stability boundary.

We then simulate the model at
\begin{equation}
	\gamma=f\gamma^*,
	\qquad
	f\in\{0.95,0.90,0.80,0.70\}.
	\label{eq:gamma-factor}
\end{equation}
Values of \(f\) closer to one correspond to parameter values closer to
this diagnostic boundary.  For each parameter configuration, we generate
\(100\) independent simulated paths of length
\[
N=n+n_{\mathrm{burn}}=25{,}000,
\]
with \(n_{\mathrm{burn}}=5{,}000\) observations discarded and
\(n=20{,}000\) retained.  Throughout the simulations, we set
\(\tau=1\) and \(\mu=0.01\), and initialize the pseudo-random number
generator with a fixed seed for reproducibility.  The design therefore
contains \(4\times6\times4=96\) parameter configurations, with \(100\)
Monte Carlo replications for each configuration.  As a bandwidth
robustness check, we also repeat the estimation using
\(m=\lfloor n^a\rfloor\) for
\(a\in\{0.55,0.60,0.70\}\).

\subsection{Results}
\label{sec:results}

Table~\ref{tab:lw} reports the Monte Carlo mean local Whittle estimates
\(\overline{\hat d}_{\log\eta^2}\) for all parameter configurations.
The estimates indicate substantial low-frequency persistence in the
simulated log-squared innovations.

\begin{table}[t]
	\centering
	\small
	\caption{Monte Carlo mean local Whittle estimates
		\(\overline{\hat d}_{\log\eta^2}\) by stability factor
		\(f=\gamma/\gamma^*\), power-law decay exponent \(p\), and shock
		coefficient \(\xi\).  The column \(\gamma^*\) is the numerically
		computed boundary induced by the diagnostic log-moment condition
		\(\ell(\gamma^*)=0\), where
		\(\ell(\gamma)=\E\log\{\rho^*(\gamma)+\xi U\}\).
		Each entry is averaged over up to \(100\) successful replications.
		Bandwidth \(m=\lfloor n^{0.65}\rfloor=624\), sample size
		\(n=20{,}000\), \(\tau=1\), and \(\mu=0.01\).}
	\label{tab:lw}
	\begin{tabular}{ccrcccc}
		\toprule
		& & & \multicolumn{4}{c}{
			\(\overline{\hat d}_{\log\eta^{2}}\) at
			\(f=\gamma/\gamma^*\)} \\
		\cmidrule(lr){4-7}
		\(p\) & \(\xi\) & \(\gamma^*\)
		& \(0.95\) & \(0.90\) & \(0.80\) & \(0.70\) \\
		\midrule
		\multirow{6}{*}{\(1.1\)}
		& \(0.10\) & \(5.4245\) & 0.437 & 0.386 & 0.295 & 0.226 \\
		& \(0.25\) & \(2.0297\) & 0.546 & 0.482 & 0.349 & 0.241 \\
		& \(0.50\) & \(0.9144\) & 0.594 & 0.515 & 0.371 & 0.243 \\
		& \(0.75\) & \(0.5463\) & 0.610 & 0.535 & 0.405 & 0.284 \\
		& \(1.00\) & \(0.3631\) & 0.653 & 0.584 & 0.458 & 0.344 \\
		& \(1.10\) & \(0.3132\) & 0.665 & 0.605 & 0.484 & 0.378 \\
		\midrule
		\multirow{6}{*}{\(1.2\)}
		& \(0.10\) & \(6.2190\) & 0.440 & 0.394 & 0.301 & 0.228 \\
		& \(0.25\) & \(2.3488\) & 0.564 & 0.490 & 0.361 & 0.247 \\
		& \(0.50\) & \(1.0719\) & 0.608 & 0.533 & 0.378 & 0.250 \\
		& \(0.75\) & \(0.6483\) & 0.645 & 0.562 & 0.418 & 0.285 \\
		& \(1.00\) & \(0.4364\) & 0.680 & 0.606 & 0.464 & 0.359 \\
		& \(1.10\) & \(0.3786\) & 0.692 & 0.628 & 0.497 & 0.398 \\
		\midrule
		\multirow{6}{*}{\(1.5\)}
		& \(0.10\) & \(8.7224\) & 0.460 & 0.407 & 0.316 & 0.237 \\
		& \(0.25\) & \(3.3610\) & 0.602 & 0.525 & 0.380 & 0.259 \\
		& \(0.50\) & \(1.5760\) & 0.668 & 0.574 & 0.419 & 0.262 \\
		& \(0.75\) & \(0.9777\) & 0.696 & 0.621 & 0.451 & 0.307 \\
		& \(1.00\) & \(0.6756\) & 0.722 & 0.654 & 0.507 & 0.375 \\
		& \(1.10\) & \(0.5926\) & 0.748 & 0.676 & 0.536 & 0.407 \\
		\midrule
		\multirow{6}{*}{\(2.0\)}
		& \(0.10\) & \(13.1723\) & 0.485 & 0.425 & 0.326 & 0.245 \\
		& \(0.25\) & \(5.1741\)  & 0.629 & 0.556 & 0.403 & 0.273 \\
		& \(0.50\) & \(2.4885\)  & 0.721 & 0.622 & 0.441 & 0.274 \\
		& \(0.75\) & \(1.5797\)  & 0.746 & 0.665 & 0.485 & 0.330 \\
		& \(1.00\) & \(1.1172\)  & 0.782 & 0.700 & 0.552 & 0.387 \\
		& \(1.10\) & \(0.9892\)  & 0.797 & 0.709 & 0.568 & 0.405 \\
		\bottomrule
	\end{tabular}
\end{table}

For fixed \((p,\xi)\), the estimates are uniformly larger when \(f\) is
closer to one.  Thus, as \(\gamma=f\gamma^*\) approaches the diagnostic
boundary, the simulated volatility proxy displays stronger persistence.
The Monte Carlo mean increases monotonically from \(f=0.70\) to
\(f=0.95\) for every \((p,\xi)\) in the parameter grid.  Moreover, the
Monte Carlo standard errors range from approximately \(0.002\) to
\(0.006\), indicating that the observed pattern is not driven by
simulation variability.  The same monotone boundary-proximity pattern
is preserved under the alternative bandwidth choices.

This pattern is consistent with the endogenous decay mechanism.  When
the scale variable \(c_n\) is large, the coefficient
\[
\beta_n
=
\left(1+\frac{\tau}{c_n}\right)^{-p}
\]
is close to one, so past squared innovations decay slowly through the
random products in the ARCH representation.

The estimates also tend to increase with the shock coefficient \(\xi\),
particularly near the diagnostic boundary.  For example, at \(p=2.0\)
and \(f=0.95\), \(\overline{\hat d}_{\log\eta^2}\) rises from \(0.485\)
when \(\xi=0.10\) to \(0.797\) when \(\xi=1.10\).  The effect of \(p\)
is weaker and less uniform, but larger values of \(p\) are generally
associated with stronger persistence when \(\xi\) is moderate or large.
For example, at \(f=0.95\) and \(\xi=1.00\), the mean estimates are
\(0.653\), \(0.680\), \(0.722\), and \(0.782\) for
\(p=1.1\), \(1.2\), \(1.5\), and \(2.0\), respectively.

These results support the persistence mechanism of the model.  Slope
matching makes the decay scale \(c_n\) endogenous, so each shock affects
both the current volatility level and the future rate of decay.
Consequently, the observable volatility proxy can exhibit
fractional-type low-frequency persistence even though the augmented
state \((X_n,c_n)\) is finite-dimensional and Markovian.

\section{Empirical study}
\label{sect:empirical}

In this section, we examine whether the two-dimensional Markov model
reproduces the persistence observed in financial volatility and compare
its in-sample and out-of-sample performance with standard volatility
models.  We consider five international equity indices and Bitcoin.

\subsection{Data}
\label{sec:data}

We use daily open-to-close log returns and five-minute realized variance
from the Oxford--Man Institute Realized Library for the S\&P~500
(\texttt{.SPX}), FTSE~100 (\texttt{.FTSE}), DAX (\texttt{.GDAXI}),
Nikkei~225 (\texttt{.N225}), and KOSPI (\texttt{.KS11}), covering
January~2000 through June~2018.

For Bitcoin, we use daily returns and realized variance constructed from
one-minute Bitstamp data covering January~2012 through January~2025.
After removing non-trading days and observations with missing values,
the samples contain between \(4{,}501\) and \(4{,}754\) daily
observations.

Returns are used to estimate the return-based volatility models.
Realized variance is used as the ex-post volatility measure in the
out-of-sample evaluation.  
The HAR-RV model additionally uses lagged realized-variance measures as predictors.

\subsection{Gaussian Quasi-Maximum Likelihood Estimation}
\label{sec:qml}

We estimate the model by Gaussian quasi-maximum likelihood, using the
Gaussian density as a working conditional density with variance
\[
\sigma_n^2=\mu+X_{n-1}.
\]

Throughout the empirical analysis, we set \(\tau=1\) and fix the
power-law exponent at \(p=1.2\).  Profile-likelihood calculations
indicate weak separate identification of \(p\) and \(\gamma\).
Fixing \(p\) yields a parsimonious three-parameter specification and
maintains a common decay-kernel shape across markets.  Unreported
estimates with \(p=2.0\) produce qualitatively similar conclusions.

The estimated parameter vector is
\[
\theta=(\mu,\xi,\gamma)^\top.
\]
For each \(\theta\), the state variables
\((X_n(\theta),c_n(\theta))\) are generated recursively from
\eqref{eq:X_update}--\eqref{eq:Y_update}.  The Gaussian
quasi-log-likelihood, up to an additive constant, is
\begin{equation}
	\ell_N(\theta)
	=
	-\frac{1}{2}
	\sum_{n=1}^{N}
	\left[
	\log \sigma_n^2(\theta)
	+
	\frac{\eta_n^2}{\sigma_n^2(\theta)}
	\right],
	\label{eq:log-likelihood}
\end{equation}
where
\[
\sigma_n^2(\theta)=\mu+X_{n-1}(\theta).
\]
The Gaussian quasi-maximum likelihood estimator is
\begin{equation}
	\hat\theta
	=
	\argmax_{\theta\in\Theta}
	\ell_N(\theta),
	\label{eq:qmle}
\end{equation}
where
\[
\Theta
=
\left\{
(\mu,\xi,\gamma)\in\mathbb R^3:
\mu>0,\ \xi>0,\ \gamma>0
\right\}.
\]

For each fitted model, we compute the diagnostic frontier \(\hat\gamma^*\) using the fitted value \(\hat\xi\), with \(p=1.2\). 
Following \eqref{eq:gamma-factor}, we define the fitted stability factor by \( \hat f = \hat\gamma/\hat\gamma^*.\) 
Values closer to one indicate greater proximity to the diagnostic stability boundary, while \(\hat f<1\) places the fitted model inside that boundary.

\begin{table}[tbp]
	\centering
	\caption{Gaussian quasi-maximum-likelihood estimates of the
		LM-GARCH model with \(p=1.2\) and \(\tau=1\).
		The stability factor
		\(\hat f=\hat\gamma/\hat\gamma^*\) measures proximity to the
		diagnostic stability boundary.  The column \(\log L\) reports
		the Gaussian quasi-log-likelihood.}
	\label{tab:base_main}
	\begin{tabular}{lrrrrr}
		\toprule
		Asset
		& \(\hat\mu\)
		& \(\hat\xi\)
		& \(\hat\gamma\)
		& \(\hat f\)
		& \(\log L\) \\
		\midrule
		S\&P 500
		& \(8.58\times10^{-6}\)
		& 0.115
		& 4.86
		& 0.898
		& 19600.5 \\
		FTSE 100
		& \(1.17\times10^{-5}\)
		& 0.121
		& 4.41
		& 0.865
		& 19456.4 \\
		DAX
		& \(1.39\times10^{-5}\)
		& 0.097
		& 5.70
		& 0.883
		& 18947.8 \\
		Nikkei 225
		& \(1.90\times10^{-5}\)
		& 0.140
		& 3.66
		& 0.836
		& 18550.6 \\
		KOSPI
		& \(8.37\times10^{-6}\)
		& 0.091
		& 6.26
		& 0.905
		& 18934.9 \\
		Bitcoin
		& \(3.32\times10^{-4}\)
		& 0.200
		& 2.20
		& 0.732
		& 13644.5 \\
		\bottomrule
	\end{tabular}
\end{table}

The estimated stability factors range from \(0.732\) to \(0.905\).
The five equity indices have values between \(0.836\) and \(0.905\),
whereas Bitcoin has a lower value of \(0.732\).  Among the fitted
models, KOSPI lies closest to the diagnostic boundary and Bitcoin lies
furthest from it.

\subsection{The Long-Memory Property}
\label{sec:longmemory}

We assess whether the fitted two-dimensional Markov model reproduces
the empirical long-memory behavior of financial volatility using three
diagnostics.

First, we estimate the memory parameter \(d\) from the data using the
local Whittle estimator applied to log-squared
returns.  The estimates range from \(0.350\) to \(0.434\) across the six
assets, lying within the stationary long-memory region
\(0<d<1/2\) and well above zero.

Second, we examine whether the fitted model reproduces the estimated
memory of each market.  For each asset, we simulate independent
trajectories from the fitted model, apply the same local Whittle
estimator to the simulated log-squared returns, and average the
resulting estimates across replications.

The model-implied estimate is nearly identical to the data-based
estimate for the Nikkei~225.  For the S\&P~500 and FTSE~100, the
differences are close to one Monte Carlo standard error.  For the DAX,
KOSPI, and Bitcoin, the fitted model produces lower memory estimates
than those observed in the data, with similar differences of about
\(0.045\).

\begin{table}[tbp]
	\centering
	\caption{Long-memory estimates at \(p=1.2\).
		\(\hat d_{\mathrm{data}}\) is the local Whittle estimate from
		log-squared returns, and \(\hat d_{\mathrm{model}}\) is the
		average estimate from independent simulations of the fitted
		model.  Monte Carlo standard errors are reported in
		parentheses.}
	\label{tab:dvalid}
	\begin{tabular}{lrrr}
		\toprule
		Asset & \(\hat d_{\mathrm{data}}\)
		& \(\hat d_{\mathrm{model}}\) & (s.e.) \\
		\midrule
		S\&P 500   & 0.399 & 0.413 & (0.013) \\
		FTSE 100   & 0.406 & 0.391 & (0.012) \\
		DAX        & 0.426 & 0.380 & (0.011) \\
		Nikkei 225 & 0.378 & 0.375 & (0.013) \\
		KOSPI      & 0.434 & 0.389 & (0.012) \\
		Bitcoin    & 0.350 & 0.305 & (0.011) \\
		\bottomrule
	\end{tabular}
\end{table}

Third, we compare the autocorrelation functions of squared returns.
Figure~\ref{fig:acf} presents the empirical autocorrelation for the
Nikkei~225 together with those implied by the fitted LM-GARCH and
GARCH\((1,1)\) models.  Both model-based autocorrelations are obtained
by simulation at the estimated parameters.

The fitted GARCH\((1,1)\) model is highly persistent, with
\(\hat\alpha+\hat\beta\approx0.98\), but its autocorrelation eventually
decays geometrically and falls below the empirical curve.  The
LM-GARCH model more closely reproduces the slow decay of the empirical
autocorrelation, although both models somewhat overstate persistence at
intermediate lags.

\begin{figure}[tbp]
	\centering
	\includegraphics[width=0.75\linewidth]{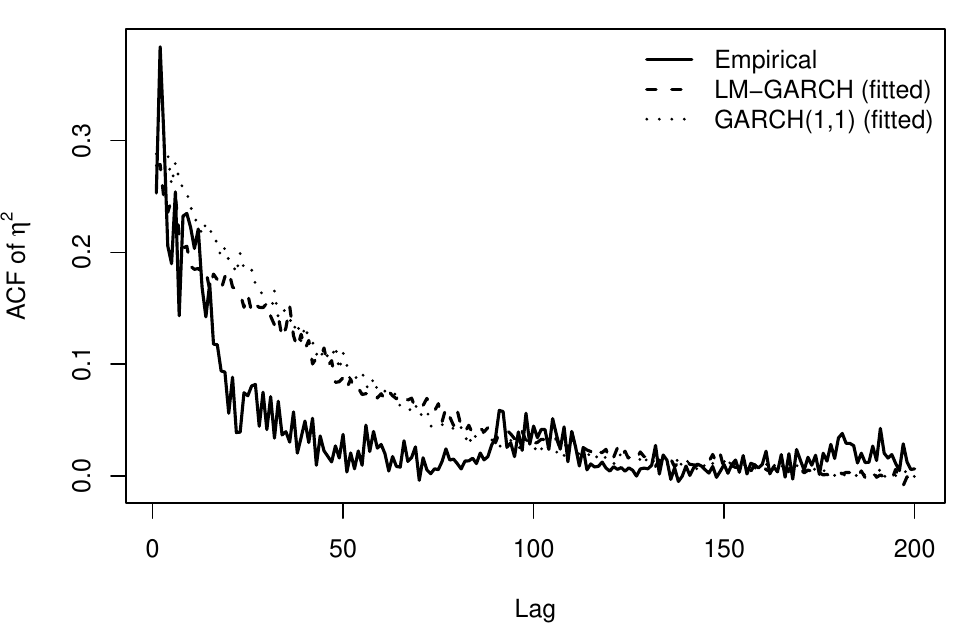}
	\caption{Autocorrelation functions of squared returns for the
		Nikkei~225 through lag \(200\): empirical, implied by the fitted
		LM-GARCH model, and implied by the fitted GARCH\((1,1)\) model.
		The model-based autocorrelations are computed from simulations
		at the estimated parameters.}
	\label{fig:acf}
\end{figure}

The stability factors and model-implied memory estimates are broadly
consistent.  Bitcoin has both the lowest stability factor and the
lowest model-implied memory estimate, whereas the equity indices lie
closer to the diagnostic boundary and generally exhibit stronger
model-implied persistence.  In unreported robustness checks,
re-estimation with \(p=2.0\) leaves the main persistence diagnostics
and cross-market conclusions qualitatively unchanged.

\subsection{In-Sample Comparison}
\label{sec:insample}

Table~\ref{tab:base_competitors} compares the in-sample
log-likelihood of the LM-GARCH model with those of
GARCH, GJR-GARCH, EGARCH, and FIGARCH.  For comparability,
all models are evaluated under the same Gaussian log-likelihood
convention using their fitted conditional variance paths.

The closest short-memory benchmark is GARCH\((1,1)\), which, like the
LM-GARCH specification considered here, is symmetric and has three
estimated variance parameters.  Their in-sample likelihoods are similar.
GARCH\((1,1)\) has a higher likelihood for four of the six markets,
although the differences are small, ranging from \(1.5\) to \(7.8\)
log-likelihood points.  LM-GARCH performs slightly better for the
FTSE~100 and Bitcoin.  This comparison indicates that the long-memory
kernel provides little systematic improvement in one-step in-sample fit
over a highly persistent short-memory specification.

GJR-GARCH and EGARCH attain substantially higher likelihoods for the
five equity indices.  These asymmetric models can capture leverage
effects that are not included in the symmetric LM-GARCH specification.
For Bitcoin, however, LM-GARCH has the highest likelihood among the
nonfractional models.

FIGARCH provides the strongest long-memory benchmark.  It has a higher
in-sample likelihood for all six markets, with differences ranging from
\(13.0\) to \(24.6\) log-likelihood points.  The two models differ,
however, in their representation of persistence: LM-GARCH uses a
two-dimensional Markov state, whereas FIGARCH relies on a fractionally
integrated infinite-order ARCH structure.
The out-of-sample comparison in Section~\ref{sec:oos} shows that
the likelihood advantage of FIGARCH translates into statistically
significant forecast improvements only for KOSPI and Bitcoin.

\begin{table}[tbp]
	\centering
	\caption{In-sample Gaussian log-likelihood comparison.  The first
		column reports the LM-GARCH log-likelihood, and the remaining
		columns report
		\(\Delta=\log L_{\mathrm{competitor}}
		-\log L_{\mathrm{LM\text{-}GARCH}}\).
		Positive values favor the competing model.  All models are
		evaluated using the same Gaussian convention applied to their
		fitted conditional variance paths.}
	\label{tab:base_competitors}
	\begin{tabular}{lrrrrr}
		\toprule
		& LM-GARCH & \multicolumn{4}{c}{\(\Delta\) relative to LM-GARCH} \\
		\cmidrule(lr){2-2}
		\cmidrule(lr){3-6}
		Asset & \(\log L\) & GARCH & GJR & EGARCH & FIGARCH \\
		\midrule
		S\&P 500   & 19600.5 &  1.5 &  96.5 & 108.6 & 13.1 \\
		FTSE 100   & 19456.4 & $-0.5$ & 101.7 & 104.2 & 13.0 \\
		DAX        & 18947.8 &  5.7 &  78.6 &  85.4 & 17.4 \\
		Nikkei 225 & 18550.6 &  7.8 &  23.4 &  50.3 & 14.4 \\
		KOSPI      & 18934.9 &  2.7 &  11.9 &  26.8 & 21.8 \\
		Bitcoin    & 13644.5 & $-3.9$ &  $-3.6$ &  $-2.3$ & 24.6 \\
		\bottomrule
	\end{tabular}
\end{table}

\subsection{Out-of-Sample Forecasting}
\label{sec:oos}

We reserve the final \(30\%\) of each sample for out-of-sample
evaluation.  One-step-ahead forecasts use an expanding estimation
window.  Model parameters are re-estimated every \(250\) trading days,
while model states and predictors are updated daily as new observations
become available.  Forecasts are evaluated against realized variance
using the QLIKE loss, which is robust to measurement error in volatility
proxies under standard conditions \citep{patton2011}.

Tables~\ref{tab:base_oos_qlike} and~\ref{tab:base_oos_dm}
report the average QLIKE losses and the corresponding
Diebold--Mariano tests, respectively.  We compare LM-GARCH with three
benchmarks: GARCH\((1,1)\), FIGARCH, and HAR-RV
\citep{corsi2009}.  GARCH\((1,1)\) is a short-memory return-based
model, FIGARCH is a long-memory return-based model, and HAR-RV
forecasts realized variance directly from its own lags.

The Diebold--Mariano tests \citep{dieboldmariano1995} use the loss
differential
\[
L_{t,\mathrm{LM\text{-}GARCH}}
-
L_{t,\mathrm{benchmark}},
\]
so a negative statistic favors LM-GARCH.  The long-run variance is
estimated using a Newey--West HAC estimator with Bartlett weights \citep{NeweyWest1987}
and
bandwidth \(\lfloor n^{1/3}\rfloor\).  Reported \(p\)-values are based
on the standard normal approximation.

\begin{table}[t]
	\centering
	\small
	\caption{Out-of-sample QLIKE losses at \(p=1.2\).
		Entries are average losses over the final \(30\%\) of each
		sample; lower values indicate greater forecast accuracy.}
	\label{tab:base_oos_qlike}
	\begin{tabular}{lrrrr}
		\toprule
		Asset
		& LM-GARCH
		& GARCH
		& HAR-RV
		& FIGARCH \\
		\midrule
		S\&P 500   & 0.3725 & 0.3766 & 0.2448 & 0.3659 \\
		FTSE 100   & 0.2809 & 0.2809 & 0.2532 & 0.2726 \\
		DAX        & 0.2205 & 0.2210 & 0.1751 & 0.2189 \\
		Nikkei 225 & 0.3737 & 0.3689 & 0.2925 & 0.3624 \\
		KOSPI      & 0.2081 & 0.2131 & 0.1408 & 0.1731 \\
		Bitcoin    & 0.3361 & 0.3358 & 0.2804 & 0.3121 \\
		\bottomrule
	\end{tabular}
\end{table}

\begin{table}[t]
	\centering
	\small
	\caption{Diebold--Mariano tests of out-of-sample forecast accuracy.
		Each statistic compares LM-GARCH with the indicated benchmark,
		using the loss differential
		\(L_{t,\mathrm{LM\text{-}GARCH}}
		-L_{t,\mathrm{benchmark}}\).
		Negative values favor LM-GARCH.
		Stars denote significance:
		\(^{*}p<0.10\), \(^{**}p<0.05\), and \(^{***}p<0.01\).}
	\label{tab:base_oos_dm}
	\begin{tabular}{lrrr}
		\toprule
		Asset
		& vs.\ GARCH
		& vs.\ HAR-RV
		& vs.\ FIGARCH \\
		\midrule
		S\&P 500
		& \(-1.29\)\dmsig{}
		& \( 8.17\)\dmsig{***}
		& \( 0.52\)\dmsig{} \\
		FTSE 100
		& \( 0.00\)\dmsig{}
		& \( 2.62\)\dmsig{***}
		& \( 1.53\)\dmsig{} \\
		DAX
		& \(-0.34\)\dmsig{}
		& \( 4.94\)\dmsig{***}
		& \( 0.46\)\dmsig{} \\
		Nikkei 225
		& \( 1.13\)\dmsig{}
		& \( 3.50\)\dmsig{***}
		& \( 1.49\)\dmsig{} \\
		KOSPI
		& \(-2.10\)\dmsig{**}
		& \( 5.69\)\dmsig{***}
		& \( 7.15\)\dmsig{***} \\
		Bitcoin
		& \( 0.14\)\dmsig{}
		& \( 2.49\)\dmsig{**}
		& \( 5.62\)\dmsig{***} \\
		\bottomrule
	\end{tabular}
\end{table}

LM-GARCH and GARCH\((1,1)\) have similar average QLIKE losses.  LM-GARCH
has a significantly lower loss only for KOSPI.  The differences are not
significant for the other five markets.

FIGARCH has a lower average loss than LM-GARCH in all six markets.  The
difference is significant for KOSPI and Bitcoin and is not significant
for the other four markets.  The results do not indicate a general
out-of-sample advantage of LM-GARCH over the return-based benchmarks.

HAR-RV has the lowest average QLIKE loss for all six markets, and the
differences relative to LM-GARCH are statistically significant.
HAR-RV is fitted directly to realized variance and uses lagged realized
variance as predictors, whereas LM-GARCH is estimated solely from
returns.  It is therefore a specialized realized-variance benchmark
rather than a directly comparable model of the conditional return
distribution.

\subsection{Discussion}
\label{sec:emp_discussion}

The empirical results yield two main conclusions.  First, the
finite-dimensional Markov model generates substantial long-memory
behavior in log-squared returns.  The model-implied memory estimates
closely match the data-based estimates for the Nikkei~225, S\&P~500,
and FTSE~100, although they are lower for the DAX, KOSPI, and Bitcoin.

Second, LM-GARCH provides out-of-sample forecast accuracy broadly
comparable to that of the return-based benchmarks.  Its QLIKE losses are similar to those of
GARCH\((1,1)\) in most markets, while FIGARCH has lower average losses
but statistically significant improvements only for KOSPI and Bitcoin.

Overall, the results support the proposed persistence mechanism while
also showing that it does not account for all of the measured memory in
every market.  Substantial volatility persistence and competitive
return-based forecasting performance are obtained within a
low-dimensional Markov structure.

\section{Conclusion}
\label{sec:conclusion}

This paper proposed a finite-dimensional Markovian volatility model
based on a latent power-law kernel.  After each squared innovation, the
kernel is updated by matching its level and slope.  The resulting
two-dimensional state process consists of a variance-level coordinate
and an endogenous memory-scale coordinate.

We established positive Harris recurrence under a joint
Foster--Lyapunov stability condition.  The proof combines the drift
condition with control-chain arguments establishing irreducibility and
petite-set regularity.  The state process therefore admits a unique
invariant probability measure.  Simulations based on local Whittle
estimation show that the model can generate substantial low-frequency
persistence in log-squared innovations, particularly near the
diagnostic stability boundary.

The empirical results support this mechanism.  The fitted model
generates strong long-memory behavior and closely matches the estimated
memory for several markets.  
Empirically, FIGARCH provides a better in-sample fit across the six markets considered.  The out-of-sample evidence is less clear-cut, as FIGARCH has a lower average QLIKE loss in every market but its advantage is statistically significant only for KOSPI and Bitcoin.  LM-GARCH also performs similarly to GARCH((1,1)) in most markets.

Thus, LM-GARCH provides competitive one-step forecast performance while
representing persistent volatility dynamics with a two-dimensional
Markov state rather than an infinite-order volatility representation.
Unlike deterministic-coefficient \(\mathrm{ARCH}(\infty)\) and
fractionally integrated models, its decay rates are random and
endogenous, while the state recursion remains low-dimensional and
computationally tractable.  The model thus provides a parsimonious
framework in which long-memory behavior, stability analysis, likelihood
estimation, and forecasting can be studied within a common Markovian
structure.

\section*{Funding}
This work was supported by the National Research Foundation of Korea (NRF)
grant funded by the Korea government (MSIT) (No.\ RS-2026-25469087).

\bibliography{Bib}
\bibliographystyle{chicago}

\begin{appendices}

\section{Algebraic rearrangement for the \(c\)-drift}
\label{app:c-drift-algebra}

For completeness, we record the algebraic rearrangement used in the proof of Lemma~\ref{lem:c-envelope}. 
Fix \((X_0,c_0)=(x,c)\), let \(z=\eta_1^2=(\mu+x)U\), and write \(r=r(c)\). 
From Eqs.~\eqref{eq:X_update} and \eqref{eq:Y_update},
\[
c_1
=
\frac{c(xr^p+\xi z)}
{xr^{p+1}+(\xi c/\gamma)z}.
\]
Therefore
\[
c_1-c
=
\frac{
	cxr^p(1-r)-\xi c(c-\gamma)z/\gamma
}
{
	xr^{p+1}+(\xi c/\gamma)z
}.
\]
Set
\[
A:=cxr^p(1-r),
\qquad
B:=\frac{\xi c(c-\gamma)}{\gamma},
\qquad
C:=xr^{p+1},
\qquad
D:=\frac{\xi c}{\gamma}.
\]
Then
\[
c_1-c=\frac{A-Bz}{C+Dz}.
\]
Using the identity
\[
\frac{A-Bz}{C+Dz}
=
-\frac{B}{D}
+
\frac{AD+BC}{D^2}\frac{D}{C+Dz},
\]
we compute
\[
\frac{B}{D}=c-\gamma
\]
and
\[
AD+BC
=
\frac{\xi cxr^p}{\gamma}
\{c(1-r)+r(c-\gamma)\}
=
\frac{\xi cxr^p}{\gamma}(c-\gamma r).
\]
Hence
\[
\frac{AD+BC}{D^2}
=
\frac{\gamma xr^p(c-\gamma r)}{\xi c}.
\]
Define
\[
k:=\frac{C}{D}
=
\frac{\gamma xr^{p+1}}{\xi c}.
\]
Since
\[
\frac{\gamma xr^p(c-\gamma r)}{\xi c}
=
\frac{k}{r}(c-\gamma r),
\]
we obtain
\[
c_1-c
=
-(c-\gamma)
+
\frac{c-\gamma r}{r}
\frac{k}{k+z}.
\]

\section{Auxiliary calculations for regularity}\label{app:regularity-details}

\subsection{Proof of Lemma~\ref{lem:attractor}}
\label{app:attractor-proof}

\begin{proof}
	Fix \((x_0,c_0)\in\mathcal S\) and \(\varepsilon\in(0,1)\). We shall
	construct \(N\in\mathbb N\) and controls
	\((u_1,\dots,u_N)\in O_U^N\) such that the corresponding controlled state
	\((X_N,c_N)\) satisfies
	\[
	X_N<\varepsilon,
	\qquad
	|c_N-\gamma|<\varepsilon.
	\]
	
	We first reduce to the case \(x_0>0\). If \(x_0=0\), then for any
	positive control \(u_1\in O_U\),
	\[
	X_1=\xi\mu u_1>0,
	\qquad
	c_1=\gamma.
	\]
	Since \(O_U\) contains positive controls arbitrarily close to zero, such
	a preliminary control can be chosen arbitrarily small. After this first
	step, the state has strictly positive \(X\)-coordinate. Relabelling the
	resulting state as the initial state, it therefore suffices to prove the
	claim for \(x_0>0\).
	
	Assume henceforth that \(x_0>0\). The construction has two phases.
	
	\medskip\noindent
	\textbf{Phase~A: Decay of \(X\) via small controls.}
	First consider the formal zero-control path \(u_k=0\). Since \(x_0>0\),
	this path is well-defined and satisfies
	\[
	X_{k+1}=X_k r(c_k)^p,
	\qquad
	c_{k+1}=c_k+\tau.
	\]
	Hence
	\[
	c_k^{(0)}=c_0+k\tau
	\]
	and
	\[
	X_k^{(0)}
	=
	x_0\prod_{j=0}^{k-1}r(c_j^{(0)})^p
	=
	x_0
	\prod_{j=0}^{k-1}
	\left(\frac{c_0+j\tau}{c_0+(j+1)\tau}\right)^p
	=
	x_0\left(\frac{c_0}{c_0+k\tau}\right)^p
	\xrightarrow[k\to\infty]{}0.
	\]
	Therefore, for any \(\eta\in(0,1)\), we can choose
	\(N_A\in\mathbb N\) such that the formal zero-control path satisfies
	\[
	X_{N_A}^{(0)}<\eta,
	\qquad
	c_{N_A}^{(0)}>1/\eta.
	\]
	For fixed \(N_A\) and \(x_0>0\), the \(N_A\)-step controlled state is
	continuous in \((u_1,\dots,u_{N_A})\) in a neighborhood of
	\(\mathbf 0\). Since \(0\in\overline{O_U}\), we may choose
	\((u_1,\dots,u_{N_A})\in O_U^{N_A}\) sufficiently close to zero so that
	the actual controlled state satisfies
	\begin{equation}
		X_{N_A}<\eta,
		\qquad
		c_{N_A}>1/\eta.
		\label{eq:phaseA-bounds}
	\end{equation}
	
	\medskip\noindent
	\textbf{Phase~B: One corrective control.}
	Starting from a state \((x,c)\) satisfying \(x<\eta\) and \(c>1/\eta\),
	apply a single positive control \(u_{N_A+1}=\delta\in O_U\). Since
	\(O_U\) contains positive controls arbitrarily close to zero, \(\delta\)
	may be chosen as small as needed. Writing \(r=r(c)\), the updated state is
	\[
	X'
	=
	xr^p+\xi(\mu+x)\delta,
	\]
	and
	\[
	c'
	=
	\frac{c\{xr^p+\xi(\mu+x)\delta\}}
	{xr^{p+1}+(\xi c/\gamma)(\mu+x)\delta}.
	\]
	Set
	\[
	A:=xr^p,
	\qquad
	B:=\xi(\mu+x)\delta.
	\]
	Then
	\[
	X'=A+B,
	\qquad
	c'=\frac{c(A+B)}{rA+(c/\gamma)B}.
	\]
	Hence
	\[
	\frac{c'}{\gamma}-1
	=
	\frac{(c-\gamma r)A}{\gamma rA+cB}.
	\]
	Since \(c\ge\gamma\), \(r\in(0,1)\), and \(B=\xi(\mu+x)\delta\), we have
	\[
	0\le
	\frac{c'}{\gamma}-1
	\le
	\frac{cA}{cB}
	=
	\frac{xr^p}{\xi(\mu+x)\delta}
	\le
	\frac{x}{\xi\mu\delta}.
	\]
	Therefore,
	\begin{equation}
		|c'-\gamma|
		\le
		\frac{\gamma x}{\xi\mu\delta}.
		\label{eq:phaseB-c}
	\end{equation}
	Moreover, since \(r^p\le1\),
	\begin{equation}
		X'
		\le
		x+\xi(\mu+x)\delta.
		\label{eq:phaseB-X}
	\end{equation}
	
	\medskip\noindent
	\textbf{Parameter selection.}
	Choose \(\delta\in O_U\) sufficiently small that
	\[
	\xi\mu\delta<\frac{\varepsilon}{2}.
	\]
	Then choose \(\eta>0\) sufficiently small that
	\[
	\eta(1+\xi\delta)
	<
	\frac{\varepsilon}{2}-\xi\mu\delta
	\qquad\text{and}\qquad
	\frac{\gamma\eta}{\xi\mu\delta}<\varepsilon.
	\]
	By Phase~A, there exist \(N_A\in\mathbb N\) and
	\((u_1,\dots,u_{N_A})\in O_U^{N_A}\) such that
	\[
	X_{N_A}<\eta,
	\qquad
	c_{N_A}>1/\eta.
	\]
	Applying the corrective control \(u_{N_A+1}=\delta\), inequalities
	\eqref{eq:phaseB-c}--\eqref{eq:phaseB-X} give
	\[
	|c_{N_A+1}-\gamma|
	<
	\varepsilon
	\]
	and
	\[
	X_{N_A+1}
	\le
	X_{N_A}+\xi(\mu+X_{N_A})\delta
	<
	\eta(1+\xi\delta)+\xi\mu\delta
	<
	\varepsilon.
	\]
	Thus \((X_{N_A+1},c_{N_A+1})\) lies in the
	\(\varepsilon\)-neighborhood of \((0,\gamma)\). Since
	\(\varepsilon\in(0,1)\) and \((x_0,c_0)\in\mathcal S\) were arbitrary,
	\[
	(0,\gamma)\in\overline{A_+((x_0,c_0))}.
	\]
	Therefore \((0,\gamma)\) is a global attracting state.
\end{proof}

\subsection{Derivation of the Jacobian determinants}
\label{app:det-rank}
Throughout this subsection, we write \((X_k,c_k)\) for the state after
\(k\) steps under the control model \eqref{eq:cm-X}--\eqref{eq:cm-Y}.
For \(x_0>0\) and \(c_0\ge\gamma\), define the two-step endpoint map
\[
H_{x_0,c_0}(u_1,u_2) := \Phi^{(2)}((x_0,c_0);u_1,u_2)
= (X_2(u_1,u_2),c_2(u_1,u_2)).
\]
Its control-to-state Jacobian is
\[
\mathbf J(x_0,c_0;u_1,u_2)
:=
D_{(u_1,u_2)}H_{x_0,c_0}(u_1,u_2).
\]
We write
\[
\mathbf J_0(x_0,c_0)
:=
\mathbf J(x_0,c_0;0,0).
\]

\begin{lemma}[Two-step Jacobian determinant]
	\label{lem:two-step-jacobian}
	Let \(x_0>0\), \(c_0\ge\gamma\), and
	\[
	r_0:=r(c_0)=\frac{c_0}{c_0+\tau}.
	\]
	Then
	\[
	\det \mathbf J_0(x_0,c_0)
	=
	\frac{\xi^{2}(\mu+x_0)(\mu+x_0 r_0^{p})\,\tau}
	{x_0 r_0^{p}\gamma^{2}}
	\mathcal Q(c_0),
	\]
	where
	\[
	\mathcal Q(c)
	:=
	\gamma(\gamma-2c-3\tau)
	+
	p(\gamma-c-\tau)(\gamma-c-2\tau).
	\]
	In particular,
	\[
	\det \mathbf J_0(x_0,c_0)\ne0
	\quad\text{if and only if}\quad
	\mathcal Q(c_0)\ne0.
	\]
\end{lemma}

\begin{proof}

Abbreviate
\[
r_0 := r(c_0) = \frac{c_0}{c_0+\tau}, \quad
r_1 := r(c_1) = \frac{c_1}{c_1+\tau},
\]
where $c_1 := c_0+\tau$ and $X_1 := x_0\,r_0^{p}$ are the states at
$\mathbf{u}=\mathbf{0}$, and
\[
A_0 := \xi(\mu+x_0), \quad
A_1 := \xi(\mu+X_1), \quad
\rho_0 := r_0 - \tfrac{c_0}{\gamma}, \quad
\rho_1 := r_1 - \tfrac{c_1}{\gamma}.
\]
Recall from \eqref{eq:X_update}--\eqref{eq:Y_update} that
\begin{equation}
	X_k = X_{k-1}\,r(c_{k-1})^{p}+\xi(\mu+X_{k-1})\,u_k,
	\quad
	c_k
	= \frac{c_{k-1}\bigl[X_{k-1}\,r(c_{k-1})^{p}+\xi(\mu+X_{k-1})\,u_k\bigr]}
	{X_{k-1}\,r(c_{k-1})^{p+1}+(\xi c_{k-1}/\gamma)(\mu+X_{k-1})\,u_k}.
	\label{eq:appendix-Psi}
\end{equation}
Note the key identity: $\left. c_k \right|_{u_k=0} = c_{k-1}+\tau$
independently of $X_{k-1}$.

\paragraph{Step~1: First-step partials at $u_1=0$.}
Differentiating \eqref{eq:appendix-Psi} with respect to $u_1$ at
$\left. \cdot \right|_{\mathbf{0}}$ and base point $(X_0,c_0)=(x_0,c_0)$:
\[
\left.\frac{\partial X_1}{\partial u_1}\right|_{\mathbf{0}} = A_0,
\qquad
\left.\frac{\partial c_1}{\partial u_1}\right|_{\mathbf{0}}
= \frac{c_0 A_0\,\rho_0}{x_0 r_0^{p+2}}
=: B_0.
\]
At $\mathbf{u}=\mathbf{0}$ the post-step state is $(X_1,c_1)=(x_0 r_0^p,\,c_0+\tau)$.

\paragraph{Step~2: Second-step partials at $u_2=0$.}
Differentiating \eqref{eq:appendix-Psi} with respect to $u_2$ at
$\left. \cdot \right|_{\mathbf{0}}$ and base point $(X_1,c_1)$:
\[
\left.\frac{\partial X_2}{\partial u_2}\right|_{\mathbf{0}} = A_1,
\qquad
\left.\frac{\partial c_2}{\partial u_2}\right|_{\mathbf{0}}
= \frac{c_1 A_1\rho_1}{X_1 r_1^{p+2}}
=: B_1.
\]
Since $\left. c_k \right|_{u_k=0}=c_{k-1}+\tau$ is independent of $X_{k-1}$,
\[
\left.\frac{\partial c_2}{\partial X_1}\right|_{\mathbf{0}} = 0,
\qquad
\left.\frac{\partial c_2}{\partial c_1}\right|_{\mathbf{0}} = 1,
\]
and
\[
\left.\frac{\partial X_2}{\partial X_1}\right|_{\mathbf{0}} = r_1^{p},
\qquad
\left.\frac{\partial X_2}{\partial c_1}\right|_{\mathbf{0}}
= \frac{p\,X_1\,r_1^{p-1}\,\tau}{(c_1+\tau)^{2}}.
\]

\paragraph{Step~3: Chain rule for the first column.}
By the chain rule,
\[
\left.\frac{\partial X_2}{\partial u_1}\right|_{\mathbf{0}}
= r_1^{p}\,A_0 + \frac{p\,X_1\,r_1^{p-1}\,\tau}{(c_1+\tau)^{2}}\,B_0,
\qquad
\left.\frac{\partial c_2}{\partial u_1}\right|_{\mathbf{0}} = B_0.
\]

\paragraph{Step~4: Jacobian determinant.}
\begin{align}
	\det \mathbf J_0 (x_0, c_0)
	&= \det\begin{pmatrix}
		r_1^{p}A_0 + \dfrac{p\,X_1 r_1^{p-1}\tau\,B_0}{(c_1+\tau)^{2}} & A_1 \\[6pt]
		B_0 & B_1
	\end{pmatrix} \notag\\[4pt]
	&= r_1^{p}A_0 B_1
	+ \frac{p\,X_1 r_1^{p-1}\tau\,B_0 B_1}{(c_1+\tau)^{2}}
	- A_1 B_0.
	\label{eq:det-three-terms}
\end{align}

\paragraph{Step~5: Closed forms for $\rho_0,\rho_1$.}
Using $r_0=c_0/c_1$ and $r_1=c_1/(c_1+\tau)$,
\[
\rho_0 = \frac{r_0\,g_0}{\gamma}, \quad
\rho_1 = \frac{r_1\,g_1}{\gamma}, \quad
g_0 := \gamma-c_0-\tau, \quad g_1 := \gamma-c_0-2\tau.
\]

\paragraph{Step~6: Simplification.}
Substituting the expressions for $B_0$, $B_1$ from Steps~1--2 and
$\rho_0$, $\rho_1$ from Step~5 into \eqref{eq:det-three-terms},
each term evaluates as follows. Using $X_1 = x_0 r_0^p$ and $c_0 = r_0 c_1$:
\begin{align*}
	r_1^p A_0 B_1
	&= \frac{c_1 A_0 A_1 g_1}{\gamma x_0 r_0^p r_1}, \\
	A_1 B_0
	&= \frac{c_1 A_0 A_1 g_0}{\gamma x_0 r_0^p}, \\
	\frac{p X_1 r_1^{p-1}\tau\, B_0 B_1}{(c_1+\tau)^2}
	&= \frac{p\tau A_0 A_1 g_0 g_1}{\gamma^2 x_0 r_0^p},
\end{align*}
where the last equality uses $(c_1+\tau)^2 = c_1^2/r_1^2$.
Combining and factoring out $A_0 A_1/(\gamma^2 x_0 r_0^p)$,
\[
\det \mathbf J_0(x_0, c_0) 
= \frac{A_0 A_1}{\gamma^2 x_0 r_0^p}
\Bigl[\gamma c_1\Bigl(\frac{g_1}{r_1} - g_0\Bigr) + p\tau\, g_0 g_1\Bigr].
\]
Since $c_1/r_1 = c_1+\tau$ and $g_0 = \gamma-c_1$, $g_1=\gamma-c_1-\tau$,
we have
\[
c_1\Bigl(\frac{g_1}{r_1} - g_0\Bigr)
= (c_1+\tau)(\gamma-c_1-\tau) - c_1(\gamma-c_1)
= \tau(\gamma - 2c_1 - \tau),
\]
giving
\[
\det \mathbf J_0 (x_0, c_0)  = \frac{\tau A_0 A_1}{\gamma^2 x_0 r_0^p}
\bigl[\gamma(\gamma - 2c_1 - \tau) + p\, g_0 g_1\bigr].
\]

\paragraph{Step~7: Final formula.}
Substituting \(A_0=\xi(\mu+x_0)\), \(A_1=\xi(\mu+x_0r_0^p)\), and
\(c_1=c_0+\tau\), we obtain
\[
\det \mathbf J_0(x_0,c_0)
=
\frac{\xi^{2}(\mu+x_0)(\mu+x_0 r_0^{p})\,\tau}
{x_0 r_0^{p}\gamma^{2}}
\mathcal Q(c_0).
\]
Since the prefactor is strictly positive for \(x_0>0\), the final
equivalence follows.

\end{proof}

\begin{proof}[Proof of Lemma~\ref{lem:anchor-rank}]
	Fix any \(a_1>0\).  Starting from \((0,\gamma)\), this control gives
	\[
	\Phi((0,\gamma);a_1)=(\xi\mu a_1,\gamma),
	\]
	whose first coordinate is strictly positive.  Moreover,
	\[
	\Phi^{(3)}((0,\gamma);a_1;u_2,u_3)
	=
	\Phi^{(2)}((\xi\mu a_1,\gamma);u_2,u_3),
	\]
	and therefore
	\[
	\mathbf J_3((0,\gamma);a_1;u_2,u_3)
	=
	\mathbf J(\xi\mu a_1,\gamma;u_2,u_3).
	\]
	
	By Lemma~\ref{lem:two-step-jacobian}, for every \(x_0>0\) and
	\(c_0\ge\gamma\),
	\[
	\det \mathbf J_0(x_0,c_0)\ne0
	\quad\Longleftrightarrow\quad
	\mathcal Q(c_0)\ne0.
	\]
	
	If \(\mathcal Q(\gamma)\ne0\), then
	\[
	\det \mathbf J_0(\xi\mu a_1,\gamma)\ne0.
	\]
	By continuity, we may choose \(a_2,a_3>0\), arbitrarily close to zero,
	such that
	\[
	\det \mathbf J_3((0,\gamma);a_1;a_2,a_3)\ne0.
	\]
	For \(a_2,a_3\) sufficiently close to zero, the endpoint
	\[
	\Phi^{(3)}((0,\gamma);a_1;a_2,a_3)
	\]
	lies in \(\mathcal S^\circ\).  Thus the first alternative holds.
	
	Suppose now that \(\mathcal Q(\gamma)=0\).  Under this condition,
	\[
	\mathcal Q(\gamma+\tau)=2\gamma(\gamma+2\tau)>0.
	\]
	For the formal zero second control, starting from
	\((\xi\mu a_1,\gamma)\), the next \(c\)-coordinate is \(\gamma+\tau\).
	By continuity, we may therefore choose \(a_2>0\), sufficiently small,
	such that
	\[
	\Phi^{(2)}((0,\gamma);a_1,a_2)=(x_2,c_2),
	\qquad
	x_2>0,
	\qquad
	\mathcal Q(c_2)\ne0.
	\]
	For this state,
	\[
	\Phi^{(4)}((0,\gamma);a_1,a_2;u_3,u_4)
	=
	\Phi^{(2)}((x_2,c_2);u_3,u_4),
	\]
	and therefore
	\[
	\mathbf J_4((0,\gamma);a_1,a_2;u_3,u_4)
	=
	\mathbf J(x_2,c_2;u_3,u_4).
	\]
	By Lemma~\ref{lem:two-step-jacobian},
	\[
	\det \mathbf J_0(x_2,c_2)\ne0.
	\]
	By continuity, we may choose \(a_3,a_4>0\), arbitrarily close to zero,
	such that
	\[
	\det \mathbf J_4((0,\gamma);a_1,a_2;a_3,a_4)\ne0.
	\]
	For \(a_3,a_4\) sufficiently close to zero, the endpoint
	\[
	\Phi^{(4)}((0,\gamma);a_1,a_2;a_3,a_4)
	\]
	lies in \(\mathcal S^\circ\).  Thus the second alternative holds.
\end{proof}

\subsection{Markov-chain regularity and recurrence}
\label{app:markov-background}
\begin{definition}[\(\psi\)-irreducibility]
	\label{def:psi-irreducibility}
	Let \(\{\mathbf X_n\}\) be a Markov chain on
	\((\mathcal S,\mathcal B(\mathcal S))\).  The chain is
	\(\phi\)-irreducible if there exists a nontrivial \(\sigma\)-finite
	measure \(\phi\) such that
	\[
	\sum_{n\ge1}P^n(\mathbf x,A)>0
	\]
	for every \(\mathbf x\in\mathcal S\) and every
	\(A\in\mathcal B(\mathcal S)\) with \(\phi(A)>0\).  Among all
	irreducibility measures there is a maximal measure \(\psi\), unique up
	to equivalence; the chain is then called \(\psi\)-irreducible.
\end{definition}

\begin{definition}[Small and petite sets]
	\label{def:small-petite}
	A set \(C\in\mathcal B(\mathcal S)\) is small if there exist
	\(m\ge1\), a nontrivial measure \(\nu\), and a constant \(\epsilon>0\)
	such that
	\[
	P^m(\mathbf x,A)
	\ge
	\epsilon \nu(A),
	\qquad
	\mathbf x\in C,\quad A\in\mathcal B(\mathcal S).
	\]
	The set \(C\) is petite if there exist a probability mass function
	\(a\) on \(\mathbb Z_+\) and a nontrivial measure \(\nu\) such that
	\[
	K_a(\mathbf x,A)
	:=
	\sum_{n=0}^{\infty}a(n)P^n(\mathbf x,A)
	\ge
	\nu(A),
	\qquad
	\mathbf x\in C,\quad A\in\mathcal B(\mathcal S).
	\]
\end{definition}

\begin{definition}[Harris recurrence and positive Harris recurrence]
	\label{def:harris-recurrent}
	Let
	\[
	P(\mathbf x,A)
	:=
	\mathbb P(\mathbf X_1\in A\mid \mathbf X_0=\mathbf x)
	\]
	be the one-step transition kernel, and let
	\(\mathbb P_{\mathbf x}\) denote the law of the chain started from
	\(\mathbf x\).  A \(\psi\)-irreducible chain is Harris recurrent if
	\[
	\mathbb P_{\mathbf x}(\tau_A<\infty)=1
	\]
	for every \(\mathbf x\in\mathcal S\) and every
	\(A\in\mathcal B(\mathcal S)\) with \(\psi(A)>0\), where
	\[
	\tau_A:=\inf\{n\ge1:\mathbf X_n\in A\}.
	\]
	It is positive Harris recurrent if it is Harris recurrent and admits
	an invariant probability measure.  For a \(\psi\)-irreducible positive
	Harris recurrent chain, this invariant probability measure is unique.
\end{definition}

\begin{theorem}[Foster--Lyapunov criterion]
	\label{thm:FL}
	Let \(\{\mathbf X_n\}\) be a \(\psi\)-irreducible Markov chain.
	Suppose there exist a petite set \(C\), a function
	\(V:\mathcal S\to[1,\infty)\) bounded on \(C\), and constants
	\(\delta>0\), \(b<\infty\) such that
	\[
	\mathbb E_{\mathbf x}[V(\mathbf X_1)]-V(\mathbf x)
	\le
	-\delta+b\mathbf 1_C(\mathbf x),
	\qquad
	\mathbf x\in\mathcal S.
	\]
	Then \(\{\mathbf X_n\}\) is positive Harris recurrent and admits a
	unique invariant probability measure.
\end{theorem}

\begin{remark}
	Theorem~\ref{thm:FL} is the drift-form consequence of
	\citet[Theorem~11.3.4]{meyn2009markov}.  Indeed, after replacing
	\(V\) by \(V/\delta\), the drift inequality becomes
	\[
	P(V/\delta)(\mathbf x)-V(\mathbf x)/\delta
	\le
	-1+\frac{b}{\delta}\mathbf 1_C(\mathbf x).
	\]
	Since \(C\) is petite and \(V\) is bounded on \(C\), the theorem implies
	positive Harris recurrence.  Uniqueness of the invariant probability
	measure follows from \(\psi\)-irreducibility.
\end{remark}

\end{appendices}
\end{document}